\documentclass[runningheads]{llncs}
\usepackage[T1]{fontenc}
\usepackage{graphicx}
\usepackage[noend,linesnumbered]{algorithm2e}
\usepackage{amsfonts}
\usepackage{booktabs}

\usepackage{prettyref}
\newrefformat{li}{line~\ref{#1}}
\newrefformat{def}{Definition~\ref{#1}}
\newrefformat{eqn}{Equation~\ref{#1}}
\newrefformat{thm}{Theorem~\ref{#1}}
\newrefformat{lem}{Lemma~\ref{#1}}
\newrefformat{alg}{Algorithm~\ref{#1}}
\newrefformat{fig}{Figure~\ref{#1}}
\newrefformat{tab}{Table~\ref{#1}}
\newrefformat{app}{Appendix~\ref{#1}}
\newcommand\pref[1]{\prettyref{#1}}

\newcommand\approachname{scapegoating size descent}
\newcommand\approach{\approachname{}}
\newcommand\Approach{Scapegoating size descent}
\newcommand\APPROACH{Scapegoating Size Descent}
\newcommand\toolname{\textsc{Shrinker}}
\newcommand\tool{\textsc{Shrinker}}

\usepackage{mathtools}
\DeclarePairedDelimiter\abs{\lvert}{\rvert}%

\usepackage{caption}
\usepackage{subcaption}

\usepackage{listings}
\let\origthelstnumber\thelstnumber
\makeatletter
\newcommand*\Suppressnumber{%
  \lst@AddToHook{OnNewLine}{%
    \let\thelstnumber\relax%
     \advance\c@lstnumber-\@ne\relax%
    }%
}

\newcommand*\Reactivatenumber[1]{%
  \setcounter{lstnumber}{\numexpr#1-1\relax}
  \lst@AddToHook{OnNewLine}{%
   \let\thelstnumber\origthelstnumber%
   \refstepcounter{lstnumber}%
  }%
}

\usepackage{tikz}

\newcommand\dashto\dashrightarrow
\usepackage{pgfplots}

\usepackage{pifont}
\tikzstyle{wbg}=[fill=white,fill opacity=0.9,text opacity=1]

\newcommand\citet\cite
\newcommand\citep\cite
\usepackage{hyperref}

\newcommand\red[1]{{\color{red}#1}}
\newcommand\redcal[1]{{\color{red}\mathcal{#1}}}

\newcommand\blue[1]{{\color{blue}#1}}
\newcommand\pur[1]{{\color{purple}#1}}

\newcommand\teal[1]{{\color{teal}#1}}

\newcommand\cStep{\blue{\mathrm{Step}}}

\newcommand\trsize[1]{\langle #1 \rangle}

\newcommand\trCanFail{\red{\mathrm{CanFail}}}
\newcommand\trMorePrecise{\red{\mathrm{MorePrecise}}}
\newcommand\trWiden{\red{\mathrm{Widen}}}
\newcommand\trStep{\red{\mathrm{Step}^\sharp}}
\newcommand\trSplit{\red{\mathrm{Split}}}
\newcommand\trgamma{\blue{\gamma^T}}
\newcommand\tra{\red{a}}
\newcommand\trb{\red{b}}

\newcommand\hSingle{\teal{\mathtt{HSingle}}}
\newcommand\hUpdate{\teal{\mathtt{HUpdate}}}
\newcommand\hDrop{\teal{\mathtt{HDrop}}}
\newcommand\hStep{\teal{\mathtt{HStep}}}
\newcommand\hgamma{\teal{\gamma^H}}

\newcommand\purcal[1]{{\color{purple}\mathcal{#1}}}

\newcommand\hCanFail{\pur{\mathtt{HCanFail}}}

\newcommand\hh{\teal{h}}

\newcommand\hStepSharp{\pur{\mathtt{HStep}^\sharp}}
\newcommand\hStepE{\pur{\mathtt{HStep}^\sharp_\exists}}
\newcommand\hMorePrecise{\pur{\mathtt{HMorePrecise}}}
\newcommand\hWiden{\pur{\mathtt{HWiden}}}
\newcommand\hSplit{\pur{\mathtt{HSplit}}}
\newcommand\hMaybeAddScapegoats{\pur{\mathtt{MaybeAddScapegoats}}}
\newcommand\hCanBlame{\pur{\mathtt{CanBlame}}}

\newcommand\hidx[2]{\blue{#1[}#2\blue{]}}

\usepackage[frozencache,cachedir=minted-cache]{minted}
\setminted{linenos,xleftmargin=8mm}

\usepackage{diagbox}

\newcommand\Vtt[1]{\texttt{\_\_VERIFIER\_#1}}

\begin{document}
\title{Automated Verification of Monotonic Data Structure Traversals in C}
%
%
\author{Matthew Sotoudeh}
\authorrunning{Matthew Sotoudeh}
\institute{Stanford University \\ \email{sotoudeh@stanford.edu}}
\maketitle              
\begin{abstract}
    Bespoke data structure operations are common in real-world C code.
    We identify one common subclass, \emph{monotonic data structure traversals}
    (MDSTs), that iterate monotonically through the structure.
    For example, \texttt{strlen} iterates from start to end of a character
    array until a null byte is found, and a binary search tree \texttt{insert}
    iterates from the tree root towards a leaf.
    We describe a new automated verification tool, \tool{}, to verify MDSTs
    written in C.
    %
    %
    \tool{} uses a new program analysis strategy called \emph{\approach{}},
    %
    which is designed to take advantage of the fact that many MDSTs produce
    very similar traces when executed on an input (e.g., some large list) as
    when executed on a `shrunk' version of the input (e.g., the same list but
    with its first element deleted).
    We introduce a new benchmark set containing over one hundred instances
    proving correctness, equivalence, and memory safety properties of dozens of
    MDSTs found in major C codebases including Linux, NetBSD, OpenBSD, QEMU,
    Git, and Musl.
    \tool{} significantly increases the number of monotonic string and list
    traversals that can be verified vs.\ a portfolio of state-of-the-art
    tools.

    \keywords{Verification \and Data Structures \and Small Model Property.}
\end{abstract}
\section{Introduction}
The C language's lack of generics and focus on performance encourages bespoke,
application-specific data structures.
%
Bugs in these data structures can threaten safety and correctness of the entire
codebase.
Hence, we desire a tool to automatically prove the correctness of such data
structure code written in C.

This paper focuses on a subclass of data structure code we call
\emph{monotonic data structure traversals} (MDSTs).
MDSTs are programs that take finitely many monotonic
sweeps through the structure, where each sweep starts at some root or head
element and moves forward on each loop iteration\footnote{We give no strict
definition of MDST; it is merely intuition guiding the design of our analysis.
Our tool remains sound (but incomplete) when applied to any C program.}.
Examples of MDSTs include classic implementations of \texttt{strlen} (start at
the first character and iterate forward until a null byte is found),
\texttt{list-search} (start at the head and iterate forward until the desired
element is found), and \texttt{bst-insert} (start at the root and iterate down
until a null pointer is found, then insert the new node).

\subsubsection{Benchmarks and Empirical Results}
Existing benchmarks sets are either not focused on MDSTs, or involve crafted
benchmarks that are not necessarily representative of real-world code.
Hence, we constructed a new program verification benchmark consisting of over
one hundred instances verifying temporal memory safety, spatial memory safety,
and correctness properties of dozens of MDSTs extracted from major C projects.
For example, one instance checks that the Linux and OpenBSD implementations of
\texttt{strcmp} return numbers with the same sign for every pair of input
strings; another checks that appending to a GNOME list increases its length by
one.

Our tool, \tool{}, nearly triples the number of string instances solved
(58~vs.~20) and more than doubles the number of list instances solved
(20~vs.~9) compared to the second-best solver.
\tool{} solves the second-most number of tree instances among the tools
evaluated, including one not solved by any other tool.
Our results indicate \tool{} would make a strong addition to a portfolio solver
and can significantly improve the state-of-the-art in verifying string and list
MDSTs.

\subsubsection{\APPROACH}
\tool{} is based on our new \emph{\approach{}} technique for verifying safety
properties, i.e., that no execution trace of a given input program crashes
(dereferences null, makes a false assertion, etc.).

Traditional program verifiers execute the program on all possible inputs at
once, tracking sets of possible program execution traces.
If fixedpoint is reached without any of the sets containing a crashing trace,
the verifier can conclude that the program is safe.
Because most programs have infinitely many possible traces, to ensure
termination the verifier must overapproximate the set of possible traces.
E.g., rather than record that there are traces where a certain variable might
have values $0$, or $2$, or $4$, \ldots, the verifier might track only that the
value is nonnegative.
While needed to make the verifier terminate, this overapproximation can make
the abstract interpreter think a crash might be possible even when it is not.

\Approach{} gives the verifier a new option: when it finds an overapproximated
trace that might crash, instead of giving up and reporting a potential error,
it is allowed to instead prove that, \emph{if} there is a reachable crashing
trace of this form, then there \emph{also} exists some strictly smaller
reachable trace that also crashes.
%
In other words, the verifier establishes that, for every possible program
execution trace either: (1)~the trace does not crash, or (2)~if the trace
crashes, then there exists some smaller trace that crashes as well.
%
Together, these facts constitute a proof by infinite descent that no trace
crashes.

Our verification tool, \tool{}, is based on these ideas.
Instead of running the program on a single abstract input, it runs two (or
more) copies of the program side-by-side, one on an abstract input~$x$ and
another on a \emph{shrunk} version~$x'$ of that input.
For example,~$x$ might be a nonempty linked list, and~$x'$ might be formed by
dropping the first node in~$x$.
Any time the abstract trace executing on~$x$ potentially crashes (reaches a
failure state), \tool{} merely needs to prove that the `scapegoat trace'
executing on~$x'$ also crashes and is smaller.

\subsubsection{Why Does It Work for MDSTs?}
The basic difficulty in automated verification of heap-manipulating programs is
that the heap can be arbitrarily large, so the verifier must track facts
involving an unknown number of values.
\Approach{} can sidestep this problem because many MDSTs do almost the same
thing when run on an input $x$ as when run on a shrunk version of that input
$x'$.
Consider a loop over a linked list: other than the very first iteration, every
subsequent iteration does exactly the same thing when run on a list $x$ as when
run on the tail list $x' = x.\mathtt{next}$ formed by dropping the first element in the
list.
Thus, \tool{} only needs to track precise facts about the finitely many memory
locations that actually differ between the executions on $x$ and $x'$.
%

%

\subsubsection{Contributions and Outline}
We make the following contributions:
\begin{enumerate}
    \item \emph{\Approach{}} framework for program analysis~(\pref{sec:Theory}).
    \item \tool{} tool for automated verification of C programs~(\pref{sec:Implementation}).
    \item Evaluation of \tool{} and multiple baseline verifiers on our new
        benchmark set of MDSTs extracted from major C projects~(\pref{sec:Evaluation}).
\end{enumerate}
\pref{sec:Background} gives preliminaries, \pref{sec:Motivating} works through
a motivating example, and~\pref{sec:Related} describes related work.
\pref{app:Limitations} describes limitations, future work, and a motivating
connection to the small scope hypothesis.
The \tool{} homepage is located at \url{https://lair.masot.net/shrinker/} and
an archival version is located at \url{https://doi.org/10.5281/zenodo.15225947}.

\section{Preliminaries and Traditional Abstract Interpretation}
\label{sec:Background}
We now formalize our notion of a program and what it means for a program
to be safe, then describe a variant of abstract interpretation,
which~\pref{sec:Theory} builds on to form \approach{} as used by \tool{}.
In addition to distinguishing names, we use \blue{blue} for concrete
states/traces and \red{red} for abstract.

\subsection{Preliminary Definitions}
\label{sec:Program}
We model the program to be verified as a transition relation on uninterpreted
states.
We make no formal assumption about what a state is, but in practice it
represents the state of the registers and heap at a given point during program
execution.
For the duration of this paper we assume a single, fixed program.
\begin{definition}
    \label{def:Program}
    We assume the program to be verified is defined by a \emph{transition
    relation} $\to$ on states: $\blue{s_1} \to \blue{s_2}$ means ``state
    $\blue{s_1}$ can transition to state $\blue{s_2}$ in one program step.''
    A \emph{trace} $\blue{s_1}, \ldots, \blue{s_n}$ is a sequence of states.
    We assume the verification conditions are specified by a (possibly
    infinite) set of \emph{initial traces}~$\blue{I}$, each of length 1, and a
    (possibly infinite) set of \emph{failure traces} $\blue{F}$.
\end{definition}
\tool{} automatically extracts the program relation $\to$, initial
traces~$\blue{I}$, and failure traces $\blue{F}$ from C code.
We use the terms `fails' and `crashes' equivalently in this paper.
We distinguish between \emph{traces} (any sequence of states) and
\emph{reachable traces} (those that can actually result from an execution of
the program).
\begin{definition}
    \label{def:Safety}
    A trace $\blue{s_1}, \ldots, \blue{s_n}$ is a \emph{reachable trace} if
    every step is a valid program transition (i.e., $\blue{s_1} \to \blue{s_2}
    \to \cdots \to \blue{s_n}$) and the singleton prefix of the trace (i.e.,
    just $\blue{s_1}$) is in the set of initial traces $\blue{I}$.
    We use $\blue{R}$ to notate the set of reachable traces.
    The program to be verified is \emph{safe} if $\blue{R} \cap \blue{F} =
    \emptyset$.
\end{definition}
Finally, we introduce notation for executing the program for one additional
timestep, i.e., extending traces by one state.
We allow nondeterminism, so the result will be a set of possible subsequent
traces.
\begin{definition}
    Given a trace $\blue{t} = (\blue{s_1}, \blue{s_2}, \ldots, \blue{s_n})$,
    $\cStep(\blue{t})$ is the possible traces reachable after
    one timestep, i.e., $\cStep(\blue{t}) = \{ (\blue{s_1}, \blue{s_2}, \ldots,
    \blue{s_n}, \blue{s_{n+1}}) \mid \blue{s_n} \to \blue{s_{n+1}} \}$.
\end{definition}

\subsection{Trace Abstractions}
\label{sec:TraceAbstraction}
Real computers are finite, but abstract interpreters must reason about a
potentially infinite number of possible program traces.
Hence, we need a finite representation of infinite sets of traces.
This representation is formalized as an \emph{abstract domain}~\cite{cousot77}.
This paper only uses abstract domains as a representation of infinite sets, and
we do not place many requirements on our abstract domain (e.g., we do not
require a Galois connection).

\begin{definition}
    An \emph{abstract trace domain} $\redcal{A^T}$ is a set of \emph{abstract
    traces} along with a \emph{concretization function} $\blue{\gamma^T}$ that
    maps abstract traces to sets of traces.
\end{definition}
The concretization function is merely used for the theoretical results: it need not
be implemented or even implementable.
We make no other formal assumption about the abstract traces.
In practice, they usually contain constraints about states in the trace, e.g.,
``the value of variable $i$ at the last state in the trace is positive,'' and
the concretization function $\blue{\gamma^T}$ gives the set of all traces
satisfying those constraints.
For the abstract interpreter to construct, manipulate, and reason about
abstract domain elements, the analysis designer must implement:
\begin{enumerate}
    \item $\red{I^\sharp}$: overapproximates the possible initial traces, i.e.,
        $\blue{I} \subseteq \blue{\gamma^T}(\red{I^\sharp})$
    \item $\trCanFail(\tra)$: tests for possible failure traces; must be true
        if $\blue{\gamma^T}(\tra) \cap \blue{F} \neq \emptyset$.
    \item $\trStep(\tra)$: applies $\cStep$ to all of the represented traces,
        i.e., for any $\blue{t} \in \blue{\gamma^T}(\tra)$ and $\blue{t'} \in
        \cStep(\blue{t})$, we have $\blue{t'} \in \trgamma(\trStep(\tra))$.
    \item $\trMorePrecise(\tra, \trb)$: true only when
        $\blue{\gamma^T}(\red{a}) \subseteq \blue{\gamma^T}(\trb)$.
    \item $\trWiden(\tra)$: introduces overapproximations to ensure termination;
        it can return anything as long as $\trMorePrecise(\tra,
        \trWiden(\tra))$.
    \item $\trSplit(\tra)$: splits a set of traces into subsets, often to
        introduce flow-, path-, or context-sensitivity into the analysis; it
        returns a list of abstract traces $\red{a'_1}, \red{a'_2}, \ldots,
        \red{a'_n}$ such that $\blue{\gamma^T}(\tra) \subseteq \bigcup_i
        \blue{\gamma^T}(\red{a'_i})$.
\end{enumerate}
%
The tool designer can instantiate this framework with different choices to
reach different points on the completeness--performance--termination tradeoff
curve, but as long as the above constraints are met soundness is guaranteed.

\subsection{Variant of Traditional Abstract Interpretation}
\label{sec:AbsInt}
\pref{alg:AbsInt} shows an automated verification algorithm based on the
traditional abstract interpretation framework.
It repeatedly calls $\trStep$ to explore the set of reachable traces.
If $\trCanFail$ reports that any one might involve a failure trace, it reports
a possible error.
Otherwise, once fixedpoint is reached, the program is guaranteed to be safe.
$\trWiden$ and $\trMorePrecise$ are used to encourage convergence, while
$\trSplit$ is used to case split abstract traces to improve precision.

\begin{algorithm}[t]
    \caption{Variant of Traditional Abstract Interpretation}\label{alg:AbsInt}
    \KwData{A program~(\pref{sec:Program}) and an abstract trace
    domain~(\pref{sec:TraceAbstraction}).}
    \KwResult{\textsc{Safe} if the program is definitely safe, or \textsc{Unknown}.}
    $\mathtt{worklist} \gets \{\red{I^\sharp}\}, \quad \mathtt{seen} \gets \{ \}$\;
    \While{$\mathrm{worklist}$ is not empty}{
        $\tra \gets \mathtt{worklist.pop()}$\;\label{li:AbsIntIterate}
        \lIf{$\trCanFail(\tra)$}{\label{li:AbsIntCheck}
            \Return{$\mathtt{Unknown}$}
        }
        $\mathtt{seen.add}(\tra)$\;\label{li:AbsIntAddToSeen}
        \ForEach{$\red{a'_i} \in \trSplit(\trStep(\red{a}))$}{\label{li:AbsIntInnerIter}
            $\red{a'_i} \gets \trWiden(\red{a'_i})$\;
            \If{there exists $\red{b} \in \mathtt{seen} \cup \mathtt{worklist}$ with $\trMorePrecise(\red{a'_i}, \trb)$}{
                \textbf{continue}\;\label{li:AISkip}
            }
            $\mathtt{worklist} \gets (\mathtt{worklist} \setminus
                \{ \trb \in \mathtt{worklist} \mid \trMorePrecise(\trb, \red{a'_i}) \}) \cup \{ \red{a'_i} \}$\;\label{li:AIPush}
        }
    }
    \Return{$\mathtt{Safe}$}\;
\end{algorithm}

The key~\pref{lem:AbsInt} guarantees that every reachable trace is represented
by some abstract trace processed on an iteration of the main loop
in~\pref{alg:AbsInt}.

\begin{lemma}
    \label{lem:AbsInt}
    If the algorithm returns \texttt{Safe}, then for any reachable trace $\blue{t}$
    there exists some abstract trace $\tra \in \mathtt{seen}$ with $\blue{t}
    \in \blue{\gamma^T}(\red{a})$.
\end{lemma}
\begin{proof}
    Induct on the length of $\blue{t} = (\blue{s_1}, \blue{s_2}, \ldots,
    \blue{s_n})$.
    The first iteration handles everything with $n = 1$.
    Otherwise, by inductive hypothesis, the prefix~$\blue{t'} = (\blue{s_1},
    \ldots, \blue{s_{n-1}})$ was added to $\mathtt{seen}$
    on~\pref{li:AbsIntAddToSeen} during some iteration.
    On that iteration, one of the $\red{a'_i}$s must have~$\blue{t} \in
    \trgamma(\red{a'_i})$, which gets added to the worklist
    on~\pref{li:AIPush}, hence processed and added to $\mathtt{seen}$ in a
    future iteration.
    Alternatively, a less-precise $\red{b}$ might have been found
    already~(\pref{li:AISkip}), but then $\blue{t} \in \trgamma(\red{b})$
    already, as desired.
    Finally, $\red{a'_i}$ might be removed from the worklist in a future
    execution of~\pref{li:AIPush}, but that only occurs if something less
    precise (hence also containing $\blue{t}$ in its concretization set) is
    added to replace it.
    \qed
\end{proof}

\begin{theorem}
    \label{thm:AbsInt}
    If \pref{alg:AbsInt} reports \texttt{Safe}, then the program is safe.
\end{theorem}
\begin{proof}
    Otherwise there must be a reachable trace $\blue{t} \in \blue{F}$, hence
    by~\pref{lem:AbsInt} there is a $\tra \in \mathtt{seen}$ with $\blue{t} \in
    \trgamma(\tra)$.
    But everything added to $\mathtt{seen}$ passed the check
    on~\pref{li:AbsIntCheck}, i.e., $\trCanFail(\tra)$ is false, contradicting
    the definition of $\trCanFail$.
    \qed
\end{proof}

\section{Motivating Example}
\label{sec:Motivating}
This section works through a concrete example showing how traditional abstract
interpretation~(\pref{sec:AbsInt}) with a precise enough abstract domain can
prove correctness of a simple heap manipulating program.
We then describe some issues that make this difficult to do reliably and sketch
how our technique, \approach{} (\pref{sec:Theory}), would approach the same
verification task.
Consider the code below, where we want to prove the
\texttt{\_\_VERIFIER\_fail()} call is unreachable.
\begin{minted}{C}
struct arr { int *data; int n_data; };
void test(struct arr arr) {
    for (int i = 0; i < arr.n_data; i++)
        arr.data[i] = 0;
    for (int i = 0; i < arr.n_data; i++)
         if (arr.data[i] != 0)
             __VERIFIER_fail(); }
\end{minted}
The example is simplified for expository purposes, ignoring techniques like
loop fusion that can solve this instance but do not generalize as well.
For space reasons we are somewhat informal; see~\pref{app:Worked} for a more
complete worked example.

\subsection{Traditional Abstract Interpretation (\pref{alg:AbsInt})}
Recall that \pref{alg:AbsInt} explores sets of possible program traces (each
set represented by an abstract trace $\red{a_i}$) and checks that none includes
a failing trace (i.e., one reaching line 7).
On termination, \pref{lem:AbsInt} guarantees that every reachable trace lies in
the concretization set of one of those abstract traces added to \texttt{seen}.
The exact behavior depends on the abstraction used, but below we have
visualized one possible result.
Each node represents an abstract trace in the final \texttt{seen} set.
An edge $\red{a_i} \to \red{a_j}$ means $\red{a_j}$ was added to the worklist
while processing $\red{a_i}$, i.e., applying $\cStep$ to a trace in
$\trgamma(\red{a_i})$ might result in a trace in $\trgamma(\red{a_j})$.
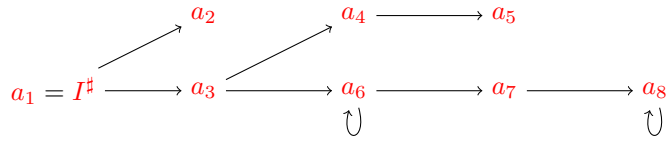
\begin{center}
    \begin{tikzpicture}
        \node (A1) at (0, 0) {$\red{a_1} = \red{I^{\sharp}}$};
        \node (A2) at (2, 1) {$\red{a_2}$};
        \node (A3) at (2, 0) {$\red{a_3}$};
        \node (A4) at (4, 1) {$\red{a_4}$};
        \node (A5) at (6, 1) {$\red{a_5}$};
        \node (A6) at (4, 0) {$\red{a_6}$};
        \node (A7) at (6, 0) {$\red{a_7}$};
        \node (A8) at (8, 0) {$\red{a_8}$};

        \draw[->] (A1) -- (A2);
        \draw[->] (A1) -- (A3);
        \draw[->] (A3) -- (A4);
        \draw[->] (A4) -- (A5);
        \draw[->] (A3) -- (A6);
        \path[->] (A6) edge [loop below] (A6);
        \draw[->] (A6) -- (A7);
        \draw[->] (A7) -- (A8);
        \path[->] (A8) edge [loop below] (A8);
    \end{tikzpicture}
\end{center}
Below, we describe each abstract trace as a set of constraints.
The concretization set consists of every trace satisfying those constraints.
We assume executing lines~3 and 5 checks the corresponding loop condition and
either executes the loop body or exits the loop.
%
\begin{itemize}
    \item $\red{a_1}$:
        About to execute line 3.
        \texttt{arr.data} points to \texttt{arr.n\_data} integers, \texttt{i=0}

    \item $\red{a_2}$:
        About to execute line 5.
        \texttt{arr.n\_data=0}, \texttt{i=0}

    \item $\red{a_3}$:
        About to execute line 3.
        \texttt{arr.n\_data>=1}, \texttt{arr.data[0]=0}, \texttt{i=1}

    \item $\red{a_4}$:
        About to execute line 5.
        \texttt{arr.n\_data=1}, \texttt{arr.data[0]=0}, \texttt{i=0}

    \item $\red{a_5}$:
        About to execute line 5.
        \texttt{arr.n\_data=1}, \texttt{arr.data[0]=0}, \texttt{i=1}

    \item $\red{a_6}$:
        About to execute line 3.
        \texttt{arr.n\_data>1}, \texttt{2<=i<=arr.n\_data}, \\
        \texttt{arr.data[0]=0}, \ldots, \texttt{arr.data[i-1]=0}

    \item $\red{a_7}$:
        About to execute line 5.
        \texttt{arr.n\_data>1}, \texttt{i=0}, \\
        \texttt{arr.data[0]=0}, \ldots, \texttt{arr.data[arr.n\_data-1]=0}

    \item $\red{a_8}$:
        About to execute line 5.
        \texttt{arr.n\_data>1}, \texttt{1<=i<=arr.n\_data}, \\
        \texttt{arr.data[0]=0}, \ldots, \texttt{arr.data[arr.n\_data-1]=0}

\end{itemize}

\pref{alg:AbsInt} can prove that traces represented by $\red{a_8}$ never reach
the crash on line 7, because all of those traces satisfy \texttt{arr.data[0] =
arr.data[1] = ... = arr.data[arr.n\_data-1] = 0}, so the \texttt{if} condition
is never taken.
This analysis, however, \textbf{requires the abstract domain to reason about
constraints involving an unknown number of memory locations}, specifically the
constraints asserting that some subset of \texttt{arr.data} entries are zero
(the ``\ldots''s in the above constraints for $\red{a_6}$, $\red{a_7}$, and
$\red{a_8}$).
If the abstract domain used was not able to represent such constraints, the
analyzer would report a false positive because it would not be able to
guarantee that the \texttt{if} condition on line 6 is not taken.

Some abstract domains can handle constraints like this~\cite{grass}, but the
larger search space makes automatically synthesizing useful invariants harder
than when restricted to only constraints that involve a finite number of memory
locations.
By contrast, the abstract traces for $\red{a_1}$ through $\red{a_5}$, which
only constrain the values of finitely many memory locations, tend to be simpler
to reason about and synthesize.
\textbf{The key goal of \approach{} is to avoid having to track precise
constraints about an unknown number of memory locations.}
Instead, we want to only track constraints about the (often finitely many)
memory locations that \emph{differ} when executing on a full input vs.\ some
related, smaller input.

\subsection{\APPROACH{}}
At a high level, our \approach{} approach also explores sets of program traces
that together account for every possible reachable trace and checks whether
they reach failure.
In fact, its handling of the traces with inputs of size 0 or 1 (i.e.,
$\red{a_1}$ through $\red{a_5}$) is essentially identical to traditional
abstract interpretation: we track constraints on the finitely many memory
locations \texttt{arr.n\_data}, \texttt{arr.data[0]}, and \texttt{i} to verify
that line 7 is never reached on any such small-sized input.
The difference comes in handling $\red{a_6}$ through $\red{a_8}$, which
represent traces that go through the loops an arbitrary (larger than 1) number
of times and which, in traditional abstract interpretation, required an
abstract domain capable of handling constraints on an unknown number of memory
locations.

Instead of directly proving that line 7 can never be reached on such traces,
\approach{} tries to prove that it can \emph{only} be reached if some smaller
trace reaches it as well.
This conditional proof is often easier to synthesize and can avoid needing to
track precise constraints on arbitrarily many memory locations, as in
$\red{a_6}$, $\red{a_7}$, $\red{a_8}$ above.
We do this by associating each such abstract trace with a \emph{scapegoat
trace} that has very closely related behavior to the \emph{primary trace} we
are concerned with.
Usually, the scapegoat trace is the result of running the program on a shrunk
version of the input for one fewer iteration of each loop: if a trace is the
result of running the program on input array $[3, 4, 10, 8]$, the scapegoat
trace might result from running the program on $[4, 10, 8]$.
For example, the equivalent of $\red{a_6}$, $\red{a_7}$, and $\red{a_8}$ are
the following:
\begin{itemize}

    \item $\red{a_6}$:
        About to execute line 3.
        \texttt{arr.n\_data>1}, \texttt{2<=i<=arr.n\_data},
        \texttt{arr.data[0]=0}. \\
        For any reachable trace satisfying those constraints, there exists
        another reachable trace (the \emph{scapegoat}) where the last state is
        identical except: \texttt{arr.data[0]} was removed,
        and both \texttt{arr.n\_data} and \texttt{i} were decremented by 1.

    \item $\red{a_7}$:
        About to execute line 5.
        \texttt{arr.n\_data>1}, \texttt{i=0}, \texttt{arr.data[0]=0}. \\
        For any reachable trace satisfying those constraints, there exists another
        reachable trace where the last state is identical except:
        \texttt{arr.data[0]} was removed and \texttt{arr.n\_data} was
        decremented by 1.

    \item $\red{a_8}$:
        About to execute line 5.
        \texttt{arr.n\_data>1}, \texttt{1<=i<=arr.n\_data},
        \texttt{arr.data[0]=0}. \\
        For any reachable trace satisfying those constraints, there exists another
        reachable trace where the last state is identical except:
        \texttt{arr.data[0]} was removed, and both \texttt{arr.n\_data} and
        \texttt{i} were decremented by 1.

\end{itemize}
Crucially, the constraints for $\red{a_8}$ imply that if some trace satisfying
the constraints of $\red{a_8}$ \emph{were} to fail (reach line 7) on the next
iteration, its corresponding scapegoat trace would \emph{also} fail.
So we have actually proved: if some input of size \texttt{arr.n\_data} leads to
a failing execution, then there is another input of strictly smaller size
\texttt{arr.n\_data - 1} that \emph{also} reaches failure.
In this way, because the size is nonnegative, we can apply \emph{proof by
infinite descent} (\pref{thm:Descent}, essentially induction on input size) to
conclude that no input causes a failure.

\section{\APPROACH}
\label{sec:Theory}
This section formalizes our \approach{} variant of abstract interpretation, as
used in \tool{}.
Traditional abstract interpretation tracks a set of possible traces.
\Approach{} modifies this framework to track a set of \emph{herds} of traces;
each herd is a \emph{primary trace} $\blue{t_1}$ along with a number of
\emph{scapegoat traces} $\blue{t_2}, \blue{t_3}, \ldots$ resulting from
different inputs or different nondeterministic choices\footnote{Apparently,
some call a group of goats a `trip' or a `tribe,' but unfortunately
`t'-starting names overload with `\underline{t}race.'}.
When the abstract interpreter thinks it might be possible for the primary trace
to have crashed, \approach{} allows the abstract interpreter to avoid
giving up by transferring the blame onto one of the scapegoat traces.
To blame a scapegoat trace, it must prove that, if the primary trace has
crashed, then the scapegoat trace has also crashed and is smaller than the
primary trace (for some definition of size; see~\pref{sec:TraceSize}).
In this way, \approach{} proves that every reachable trace either does not
crash, \emph{or}, if it does crash, then there is some strictly smaller
reachable trace that also crashes.
If the size measure is \emph{well-founded} (e.g., natural numbers), proof by
infinite descent~(\pref{thm:Descent}) ensures that no traces crash, i.e., the
program is safe.
A detailed worked example is provided in~\pref{app:Worked}.

\subsection{Trace Sizes and Infinite Descent}
\label{sec:TraceSize}
We assume the tool designer provides a measure of the \emph{size} of a trace.
%
\begin{definition}
    A \emph{trace size function} $\trsize{\cdot}$ maps traces $\blue{t}$ to a
    \emph{size} $\trsize{\blue{t}} \in \mathbb{N}$.
\end{definition}
\tool{} currently uses a measure of size that essentially counts the number of
allocated items on the heap.
But \tool{} is fairly robust to the exact measure of size; we expect that the
number of bytes allocated or even the length of the trace itself would work as
well~(see \pref{app:TraceSize} for further discussion).
The choice of trace size affects only completeness, not soundness.
Since our trace sizes are natural numbers, we can use \emph{proof by infinite
descent}.
\begin{theorem} (Proof by Infinite Descent)
    \label{thm:Descent}
    Let $P(n)$ be any statement parameterized by $n \in \mathbb{N}$.
    Suppose that whenever $P(n)$ is false, there exists $m \in \mathbb{N}$ such
    that $m < n$ and $P(m)$ is also false.\footnote{Inparticular, this implies
    $P(0)$ is not false, as $0$ has no predecessor in $\mathbb{N}$.}
    Then, $P(n)$ is true for all $n \in \mathbb{N}$.
\end{theorem}
In fact, our results generalize to any measure of size as long as the comparison
relation admits no infinite descending chains, i.e., there exists no infinite
sequence of traces $\trsize{\blue{t_1}} > \trsize{\blue{t_2}} > \ldots$.

\subsection{Trace Herds and Abstract Trace Herds}
Rather than tracking sets of traces, \approach{} tracks sets of \emph{herds} of
traces.
In addition to distinguishing names, we will color herds in \teal{teal} and
abstract herds in \pur{purple}.
\begin{definition}
    A \emph{herd} is an ordered tuple of traces. If $\teal{h}$ is a herd, then
    $\abs{\teal{h}}$ is the size of $\teal{h}$ and $\blue{h[}1\blue{]}, \ldots,
    \blue{h[}\abs{\teal{h}}\blue{]}$ are all traces.
    We use $\teal{h} = (\blue{h[}1\blue{]}, \ldots, \blue{h[}n\blue{]})$ to
    indicate a herd of size $\abs{\teal{h}} = n$.
    We call $\blue{h[}1\blue{]}$ the \emph{primary trace} of the herd, and
    $\blue{h[}2\blue{]}, \ldots, \blue{h[}n\blue{]}$ the \emph{scapegoat
    traces} of the herd.
\end{definition}
The size $\abs{\teal{h}}$ of a herd is just the number of traces in it; there
is no relation to the size of an individual trace in the herd.
The first trace in the herd plays the role of the primary trace we are trying
to prove things about; in fact, \approach{} restricted to size-1 herds is
identical to traditional abstract interpretation~(\pref{alg:AbsInt}).
$\hSingle$ constructs a singleton herd from a trace:
\begin{definition}
    Given a trace $\blue{t}$, we define $\hSingle(\blue{t}) =
    (\blue{t})$ to be the herd of size 1 consisting of only the primary trace
    $\blue{t}$ and no scapegoats.
    Given a set of traces $\blue{T}$, we overload $\hSingle(\blue{T})$ to mean the set
    $\{ \hSingle(\blue{t}) \mid \blue{t} \in \blue{T} \}$.
\end{definition}
Given a herd, we need notation for modifying individual traces in the herd.
\begin{definition}
    Given a herd $\teal{h}$, index $i$, and trace $\blue{t}$, we define the
    \emph{update-index} function
    $
        \hUpdate(\teal{h}, i, \blue{t})
        = (\blue{h[}1\blue{]}, \ldots, \blue{h[}i-1\blue{]}, \blue{t}, \blue{h[}i+1\blue{]}, \ldots, \blue{h[}n\blue{]}).
    $
    We also define the \emph{drop-index} function
    $
        \hDrop(\teal{h}, i) = (\blue{h[}1\blue{]}, \ldots, \blue{h[}i-1\blue{]}, \blue{h[}i+1\blue{]}, \ldots, \blue{h[}n\blue{]}).
    $
\end{definition}
Finally, we need to be able to extend traces in the herd by executing the
program on that trace for one timestep.
\begin{definition}
    Given a herd $\teal{h}$, we define the function
    $
        \hStep(\teal{h}, i)
        = \{ \hUpdate(\teal{h}, i, \blue{t}) \mid \blue{t} \in \cStep(\blue{h[}i\blue{]}) \},
    $
    or $\hStep(\teal{h}, i) = \teal{h}$ if $i$ is not between $1$ and $\abs{\teal{h}}$.
\end{definition}

\subsection{Herd Abstractions}
\label{sec:HerdAbstraction}
As with~\pref{alg:AbsInt}, the tool designer must provide an abstract domain to
represent sets of herds along with a number of abstract domain operations.
\begin{definition}
    An \emph{abstract herd domain} $\purcal{A^H}$ is a set of \emph{abstract
    herds} along with a \emph{concretization function} $\hgamma$ that maps
    abstract herds to sets of herds.
\end{definition}
Once again, we make no assumptions about what the abstract herds are, and the
concretization function is only needed for the theoretical results; it does not
need to be implementable.
The tool designer does still need to provide a number of operations for working
with abstract herds.

The following operations are essentially lifting ones we used
in~\pref{sec:TraceAbstraction} to abstract herds rather than abstract traces:
\begin{enumerate}
    \item $\pur{I^{H\sharp}}$: represents all herds where the primary
        trace is a starting trace, i.e., $\hSingle(\blue{I}) \subseteq
        \hgamma(\pur{I^{H\sharp}})$.
    \item $\hCanFail(\pur{a})$: must be true if any of the herds has a failing
        primary trace, i.e., true whenever there exists $\hh \in
        \hgamma(\pur{a})$ with $\hidx{h}{1} \in \blue{F}$.
    \item $\hStepSharp(\pur{a})$: runs the primary trace in each herd forward,
        i.e., for every $\hh \in \hgamma(\pur{a})$ and $\teal{h'} \in
        \hStep(\hh, 1)$, we have $\teal{h'} \in
        \hgamma(\hStepSharp(\pur{a}))$.
    \item $\hMorePrecise(\pur{a}, \pur{b})$: true only when
        $\hgamma(\pur{a}) \subseteq \hgamma(\pur{b})$.
    \item $\hWiden(\pur{a})$: returns anything as long as
        $\hMorePrecise(\pur{a}, \hWiden(\pur{a}))$.
    \item $\hSplit(\pur{a})$: returns abstract herds
        $\pur{a'_1}, \pur{a'_2}, \ldots, \pur{a'_n}$ such that
        $\hgamma(\pur{a}) \subseteq \bigcup_i \hgamma(\pur{a'_i})$.
\end{enumerate}
The following are new operations only used in \approach{}:
\begin{enumerate}
    \item $\hMaybeAddScapegoats(\pur{a})$ adds candidate scapegoat
        traces to the herd.
        These candidate scapegoat traces must be reachable traces whenever the
        primary trace is reachable.
        In practice, the scapegoat traces are constructed by modifications to
        the input in the main trace, e.g., dropping the first node in a linked
        list input argument.
        Formally, for every herd $\hh \in
        \hgamma(\pur{a})$, either
        \begin{enumerate}
            \item $\blue{h[}1\blue{]}$ is not a reachable trace, or
            \item there exists $\teal{h'} \in
                \hgamma(\hMaybeAddScapegoats(\pur{a}))$ extending $\teal{h}$, i.e., where $\teal{h'} =
                (\blue{h[}1\blue{]}, \ldots, \blue{h[}\teal{\abs{h}}\blue{]},
                \blue{t_1}, \ldots, \blue{t_n})$ and $\blue{t_1}, \ldots,
                \blue{t_n}$ are all reachable traces.
        \end{enumerate}

    \item $\hStepE(\pur{a}, i)$ runs the $i$th trace in the herd (must be a
        scapegoat, i.e., $i > 1$) forward for one timestep.
        If multiple subsequent states are possible, e.g., because a
        nondeterministic choice was made by the program, it may pick any one of
        the choices (we only need to guarantee the existence of at least one
        scapegoat trace, not analyze all possible scapegoat traces).
        Alternatively, it may drop a scapegoat trace, e.g., if no successors
        exist.
        Formally, for any~$\hh \in \hgamma(\pur{a})$, either:
        \begin{enumerate}
            \item there exists a $\teal{h'} \in \hStep(\hh, i)$ such that
                $\teal{h'} \in \hgamma(\hStepE(\pur{a}, i))$, or
            \item $\hDrop(\hh, i) \in \hgamma(\hStepE(\pur{a}, i))$.
        \end{enumerate}

    \item $\hCanBlame(\pur{a}, i)$ determines whether we can blame the
        $i$th scapegoat in the herd, i.e., whether it represents a trace that
        reached failure on a strictly smaller input.  Formally, returns true
        only if for every $\hh \in \hgamma(\pur{a})$ both:
        \begin{enumerate}
            \item $\trsize{\blue{h[}i\blue{]}} < \trsize{\blue{h[}1\blue{]}}$, and
            \item $\blue{h[}i\blue{]} \in \blue{F}$.
        \end{enumerate}
\end{enumerate}

\subsection{Algorithm}
\begin{algorithm}[t]
    \caption{\Approach{}.
    For soundness, \texttt{StepperHeuristic} may return any sequence of numbers
    greater than 1.
    }
    \label{alg:OurAlgo}
    \KwData{A program~(\pref{sec:Program}) and an abstract herd
    domain~(\pref{sec:HerdAbstraction}).}
    \KwResult{\textsc{Safe} if the program is definitely safe, or \textsc{Unknown}.}
    $\mathtt{worklist} \gets \{\pur{I^{H\sharp}}\}, \quad \mathtt{seen} \gets \{ \}$\;
    \While{$\mathtt{worklist}$ is not empty}{
        $\pur{a} \gets \mathtt{worklist.pop()}$\;\label{li:RIIIterate}

        \If{$\hCanFail(\pur{a})$ and there is no $i>1$ with $\hCanBlame(\pur{a}, i)$}{\label{li:OurCheck}
            \Return{$\mathtt{Unknown}$}\;
        }
        $\mathtt{seen.add}(\pur{a})$\;\label{li:OurSeen}
        \For{$\pur{a'_i} \in \hSplit(\hStepSharp(\pur{a}))$}{\label{li:RIIInnerIter}
            $\pur{a'_i} \gets \pur{\mathtt{MaybeAddScapegoats}^\sharp}(\pur{a'_i})$\;
            \lFor{$j \gets \mathtt{StepperHeuristic}()$ with $j > 1$}{
                $\pur{a'_i} \gets \hStepE(\pur{a'_i}, j)$
            }
            $\pur{a'_i} \gets \hWiden(\pur{a'_i})$\;
            \If{there exists $\pur{b} \in \mathtt{seen} \cup \mathtt{worklist}$ with $\hMorePrecise(\pur{a'_i}, \pur{b})$}{
                \textbf{continue}\;\label{li:OurSkip}
            }
            $\mathtt{worklist} \gets (\mathtt{worklist} \setminus
                \{ \pur{b} \in \mathtt{worklist}
                \mid \hMorePrecise(\pur{b}, \pur{a'_i}) \}) \cup \{ \pur{a'_i} \}$\;\label{li:OurPush}
        }
    }
    \Return{$\mathtt{Safe}$}\;
\end{algorithm}

The \approach{} algorithm is presented in~\pref{alg:OurAlgo}.
It is almost identical to~\pref{alg:AbsInt}, except (1) the verifier can add
and step forward scapegoat traces arbitrarily, i.e., according to the
heuristics $\hMaybeAddScapegoats$ and \texttt{StepperHeuristic}; and (2) the
verifier can ignore a possibly failing herd if one of the scapegoat traces can
be successfully blamed.

\subsubsection{Correctness Proof}
In traditional abstract interpretation,~\pref{lem:AbsInt} guaranteed that every
reachable trace was in the concretization set of some seen abstract trace.
But in \approachname{}, the abstract elements represent sets of herds, not
sets of traces, so we need to be more precise about what we mean when we say an
abstract herd seen by the algorithm accounts for a given trace.
We will say that an abstract herd $\pur{a}$ accounts for a trace
$\blue{t}$ if there is some herd in the concretization set of $\pur{a}$ where
(i) $\blue{t}$ is the primary, and (ii) everything in the herd is reachable.
The following definition makes this precise.
\begin{definition}
    \label{def:AccountsFor}
    An abstract herd $\pur{a}$ \emph{accounts for} a trace $\blue{t}$ if there
    exists some herd $\teal{h} \in \hgamma(\pur{a})$ with $\teal{h[}1\teal{]} =
    \blue{t}$ and all of the traces in $\teal{h}$ are reachable.
    In that case, we say $\teal{h}$ \emph{witnesses} that $\pur{a}$ accounts
    for $\blue{t}$.
\end{definition}
\pref{lem:Scapegoats} observes that none of the new, scapegoat-only operations
that~\pref{alg:AbsInt} performs can decrease the set of traces accounted for.
\begin{lemma}
    \label{lem:Scapegoats}
    Let $\pur{a}$ be an abstract herd and suppose $\blue{t}$ is a trace
    accounted for by $\pur{a}$.
    Then, after computing
        $\pur{a'} = \hWiden(\pur{a})$,
        $\pur{a'} = \hMaybeAddScapegoats(\pur{a})$, or
        $\pur{a'} = \hStepE(\pur{a}, i)$ for any $i > 1$,
    $\blue{t}$ is still accounted for by $\pur{a'}$.
\end{lemma}
\begin{proof}
    Let $\teal{h} \in \hgamma(\pur{a})$ be the herd witnessing that $\pur{a}$
    accounts for $\blue{t}$~(\pref{def:AccountsFor}).
    We must show there exists some $\teal{h'}$ witnessing that $\pur{a'}$
    accounts for $\blue{t}$ as well.
    For $\pur{a'} = \hWiden(\pur{a})$, the definition guarantees that
    $\hgamma(\pur{a}) \subseteq \hgamma(\pur{a'})$ so in particular we can take
    $\teal{h'}$ to be $\teal{h}$.
    For $\pur{a'} = \hMaybeAddScapegoats(\pur{a})$, we know that a valid
    $\teal{h'}$ exists by the requirements on case (b) of the definition of
    $\hMaybeAddScapegoats$ in~\pref{sec:HerdAbstraction}.
    For $\pur{a'} = \hStepE(\pur{a}, i)$ with $i > 1$, we know that there
    exists some $\teal{h'} \in \hgamma(\pur{a'})$ such that either (i) a
    $\teal{h'} \in \hStep(\teal{h}, i)$, or (ii) $\teal{h'} \in \hDrop(\teal{h},
    i)$.
    In either case, $\teal{h'}$ satisfies the desired conditions.
    \qed
\end{proof}
We can now prove the equivalent of~\pref{lem:AbsInt} in almost exactly the same
way as~\pref{sec:AbsInt}, except replacing ``$\blue{t} \in \trgamma(\red{a})$'' with
``$\pur{a}$ accounts for $\blue{t}$.''
\begin{lemma}
    \label{lem:OurAlgo}
    If~\pref{alg:OurAlgo} returns \texttt{Safe}, then for any reachable trace
    $\blue{t}$ there exists an abstract herd $\pur{a} \in \mathtt{seen}$
    accounting for $\blue{t}$.
\end{lemma}
\begin{proof}
    Induct on the length of $\blue{t} = (\blue{s_1}, \blue{s_2}, \ldots,
    \blue{s_n})$.
    For the base case, if $n = 1$ then it is accounted for by $\pur{a}$
    when~\pref{li:OurSeen} is reached on the first iteration.
    Otherwise, by inductive hypothesis, the prefix $\blue{t'} = (\blue{s_1},
    \ldots, \blue{s_{n-1}})$ was accounted for by some~$\pur{a}$ added to
    $\mathtt{seen}$ on~\pref{li:OurSeen} during some iteration.
    On that iteration, one of the~$\pur{a'_i}$s must account for $\blue{t}$ and
    get added to the worklist, and hence processed and added to $\mathtt{seen}$
    in a future iteration (\pref{lem:Scapegoats} guarantees that it still
    accounts for $\blue{t}$ even after executing $\hMaybeAddScapegoats$,
    $\hStepE$, and $\hWiden$ in the inner loop).
    Alternatively, a less-precise $\pur{b}$ might have been found
    already~(\pref{li:OurSkip}), but then $\blue{t}$ will be accounted for by
    $\pur{b}$ already, as desired.
    Notably, it is possible for $\pur{a'_i}$ to be removed from the worklist in
    a future execution of~\pref{li:OurPush} but that only occurs if something
    less precise (hence also accounting for $\blue{t}$) is added to replace it.
    %
    %
    \qed
\end{proof}

\begin{lemma}
    \label{lem:Descent}
    If~\pref{alg:OurAlgo} reports \texttt{Safe}, then for every reachable trace
    $\blue{t}$ either~$\blue{t} \not\in \blue{F}$ or there is another reachable
    trace~$\blue{t'} \in \blue{F}$ with $\trsize{\blue{t'}} <
    \trsize{\blue{t}}$.
\end{lemma}
\begin{proof}
    From~\pref{lem:OurAlgo} an abstract herd $\pur{a}$ was added to
    \texttt{seen} with some herd~$\hh \in \hgamma(\pur{a})$ having primary
    trace~$\blue{t}~=~\blue{h[}1\blue{]}$ and reachable scapegoats
    $\blue{h[}2\blue{]}$, \ldots, $\blue{h[}n\blue{]}$. But
    for~\pref{alg:OurAlgo} to return \texttt{Safe}, it must have passed the
    $\hCanFail$ and~$\hCanBlame$ check on~\pref{li:OurCheck}, i.e., either
    $\blue{t} \not\in \blue{F}$ or some scapegoat $\blue{t'} =
    \blue{h[}i\blue{]}$ is smaller and also fails, i.e., $\blue{t'} \in
    \blue{F}$ and~$\trsize{\blue{t'}} < \trsize{\blue{t}}$ as desired.
    \qed
\end{proof}

\begin{theorem}
    \label{thm:OurAlgo}
    If~\pref{alg:OurAlgo} reports \texttt{Safe}, then the program is safe.
\end{theorem}
\begin{proof}
    Using~\pref{lem:Descent} we can apply proof by infinite
    descent~(\pref{thm:Descent}) to the claim ``no reachable trace of size $n$
    is in $\blue{F}$'' and conclude that no reachable trace (of any size) is in
    $\blue{F}$, i.e., the program is safe.
\end{proof}

\section{The \texttt{Shrinker} Tool}
\label{sec:Implementation}
This section describes our \approach{} implementation \tool{}.
The \tool{} homepage is located at \url{https://lair.masot.net/shrinker/} and
an archival version with benchmarks and baseline tools is located at
\url{https://doi.org/10.5281/zenodo.15225947}.

\subsection{User Interface}
Verification goals are provided by the user to \toolname{} as a C file defining
a special \texttt{test} function.
This function may take parameters, and it may call the special methods
\Vtt{ignore()} and \Vtt{fail()}.
\toolname{} tries to prove that no input to \texttt{test} produces a trace that
calls \Vtt{fail()} without first calling \Vtt{ignore()}.
It is useful to wrap those methods in \texttt{assume(X)} and \texttt{assert(X)}
macros that check a condition before ignoring or failing.
This lets the user express program-specific correctness properties without
requiring the user to learn a complicated logical notation.
\tool{} automatically instruments pointer operations to check memory safety
properties.
Optional overflow checking can also be implemented by instrumentation.
We also support \Vtt{nondet\_type()} methods to get nondeterministic values.
We do not support VLAs or explicit C array types, but the user can specify that
an input pointer points to an arbitrarily sized array of
items~(\pref{app:ArrayTypes}).

\paragraph{Subset of C Supported}
\toolname{} supports a usable subset of C including structs, pointers, loops,
nonrecursive and tail-recursive function calls, and the standard integer types.
We throw an error immediately upon seeing unsupported parts of C, such as union
types, function pointers, and array types.
%

\paragraph{Assumption that Inputs Point to Disjoint Heaps}
For linked structures, \toolname{} verifies the correctness condition under the
additional assumption that the inputs to the function point to disjoint,
acyclic heaps, i.e., we only consider tree-shaped input structures.
This only affects \emph{inputs} to the function; the function can itself modify
the input into any form it wishes and call other functions with cyclic inputs.
This is how, e.g., we verify doubly and cyclicly linked structures: the test
harness first rewrites the acyclic input into a cyclic list and only then is
the relevant operation performed.
See~\pref{app:DisjointHeapsInterior} for more details.

\paragraph{Array and String Inputs}
Array inputs are specified by a struct type having two fields, one integer
length field named \texttt{n\_X} and one pointer field named \texttt{X}
(see~\pref{app:ArrayTypes} for an example).
\toolname{} verifies the program under the assumption that all instances of
such structs reachable from the input arguments are initialized with a
nonnegative value for the \texttt{n\_X} field and an allocated memory region of
exactly \texttt{n\_X} items pointed to be the \texttt{X} field.
String inputs can be specified by declaring an array-of-chars input, then
having the test harness iterate over it at the start of the test harness and
call \Vtt{ignore()} if it is not a properly formatted string.

\subsection{Tool Organization and Trusted Code Base}
\tool{} is unusually small, having fewer than 7 thousand lines total of C and
Python code, with no runtime dependence on third-party libraries.
A small trusted codebase can improve maintainability and confidence in its
soundness.

\tool{} includes a parser written in Python that lowers C code to a simple
intermediate representation.
Each operation in this intermediate representation has corresponding
implementations of abstract transformers (i.e., $\hStepSharp$ and~$\hStepE$)
that together encode the semantics of the program.
To keep the implementation manageable, we do not support array types, unions,
or function pointers.
We also inline all function calls, hence we only support nonrecursive and
tail-recursive function calls.
%
%
When unsupported syntax is encountered, we provide a line number and
descriptive error message to the user.

The remainder of the tool is organized as described in~\pref{sec:Theory}, using
an abstract domain we wrote with core operations implemented in C for
efficiency.
One other major optimization was to parallelize the tool (see~\pref{app:Parallel}).

\subsection{Abstract Herd Domain}
We represent abstract herds as constraints on the values of memory locations in
different states in each trace.
Constraints can relate valuess across different states, memory locations, and
traces in the herd.
They can also constrain the possible values of trace metadata, e.g., what the
`program counter' (next instruction to be executed) is.
Examples of constraints include:
\begin{itemize}
    \item ``The value of $i$ in the first state of trace 1 is one less than the
        value of $i$ in the first state of trace 2,''
    \item ``The value of $j$ in the last state of trace 5 is positive,''
    \item ``If $x$ is positive in the second-to-last state of trace 1, then the
        program counter in the last state of trace 1 is instruction 10 in the
        intermediate representation of the program; otherwise it is instruction
        20.''
\end{itemize}
The abstraction can be queried, e.g., to ask questions like:
\begin{itemize}
    \item ``Can $j$ be nonzero in the last state of trace 5?''
    \item ``What are the possible program counter values in the last
        state of trace 1?''
\end{itemize}

\subsubsection{Memory Abstraction}
The above informal examples refer to local variables in the program.
But to verify heap-manipulating code, we need the ability to refer to locations
in the heap.
This is done using heap addresses, i.e., constraints can refer to a term
representing ``the value at memory address X in the $i$th state of trace $j$.''
We use a memory abstraction inspired by the three-valued logic
analyzer~\cite{tvla}.
We track facts about two types of locations in memory: either \emph{concrete}
locations that represent a specific address in memory (e.g., the first node in
a linked list), or a \emph{summary} location that represents multiple addresses
in memory (e.g., all of the remaining nodes in the list).
We implement this with a two-level memory abstraction: every memory location
has both a \emph{major} and \emph{minor} address, and summary locations refer
to a group of memory locations that share the same major address.
We implement linked structures of arbitrary size by adding a summary node to
represent all nodes beyond a certain depth.
We prevent this addressing scheme from leaking into the program, e.g., by
disallowing the casting of non-NULL pointers into integers.
We implement arrays by giving all entries in the array a single major address,
introducing concrete nodes for the first few entries in the array, and then
introducing a summary node to represent the remainder of the array (the number
of entries to make concrete nodes for is a user-configurable parameter).

\subsubsection{Numerical Reasoning}
\label{sec:NumericalAbstraction}
\tool{} only adds a small number of constraints, e.g., applying $\hStep$ to a
program about to execute a line \texttt{i=j;} will add a constraint saying that
\texttt{i} in the last state has the same value as \texttt{j} in the
second-to-last state (along with other constraints asserting that no other
memory location has changed).
These often imply additional implicit constraints, e.g., if we also know
that \texttt{j>k} in the second-to-last state, we can infer \texttt{i>k} in the
last state (after applying \texttt{i=j}).
To make such inferences, we wrote a standard integer difference logic (IDL)
solver to determine all relations implied by constraints of the form $x - y
\leq c$ where $x$ and $y$ are terms and $c$ is a constant
offset~\cite{clrs_idl}.
All terms in the state (even nonnumerical ones) are represented in the IDL
solver; boolean terms can be encoded as 0 for false and 1 for true.
Additional rules infer basic numerical and logical properties, e.g., when $a =
b$ and some fact $F(a)$ is true, the fact $F(b)$ can be deduced.
Additional checks are added to properly model unsigned \texttt{int} overflow
and casting behavior according to the C standard, even though the underlying
solver treats all variables as mathematical integers.

\subsection{Widening ($\hWiden$)}
$\hWiden(\pur{a})$ is implemented by dropping constraints in $\pur{a}$
heuristically.
\tool{} only widens at loop iteration points, and only once the loop has been
unwound for a certain (user-controlled) number of times or a summary region has
been accessed during an earlier iteration of that loop.
%
%
We first remove all constraints referring to anything other than the very first
and last states in the trace.
We then search through the \texttt{seen} and \texttt{worklist} lists for other
abstract trace herds with the same abstract path (essentially, about to execute
the same line; see~\pref{app:AbstractPaths} for more details), and weaken any
constraints that are not shared by all of those herds to just store the sign of
the difference (e.g., if one implies $a - b < -4$ and another implies $a - b <
-7$, we weaken to the constraint $a - b < 0$).
Because there are only finitely many major addresses in our memory abstraction,
we exempt constraints describing the major address portion of a pointer and
instead try to track the precise list of all possibilities (a threshold is used
to overapproximate if even this gets too large).

\subsubsection{Scapegoating and Other Heuristics}
For space reasons, details of our other heuristics, e.g., for adding and
stepping scapegoats, are deferred to~\pref{app:Tool}.
Briefly stated, we keep the analysis precise up to a certain unrolling depth
for each loop.
Then, we add scapegoats corresponding to traces formed by dropping elements
from the input structure (e.g., the first element of an array or list).
We step those scapegoats forward until the two loops come in-sync, i.e.,
pointers to input nodes point to the same thing in the primary and the
scapegoat, and integer loop indices into arrays differ by one.
Then we step the scapegoats in lockstep with the primary trace until the loop
is exited.

\section{Evaluation}
\label{sec:Evaluation}
This section describes our benchmark set and empirical evaluation.
Experiments were run on a Debian 12 virtual machine on an Intel i9-13900.
Benchexec was used to limit RAM to 32GB and wallclock to 3 hours per
instance--tool pair.

\subsection{Benchmarks}
\label{sec:Benchmarks}
We collected a set of benchmarks verifying correctness, memory safety, and
equivalence properties of dozens of MDSTs from major real-world C projects.
The full list of projects we extracted data structure traversals from are:
Linux~\citep{linux}, NetBSD~\citep{netbsd}, OpenBSD~\citep{openbsd},
Musl~\citep{musl}, GLib~\citep{glib}, QEMU~\citep{qemu}, Redis~\citep{redis},
Zsh~\citep{zsh}, Git~\citep{git}, and GLibC~\citep{glibc}.
We divided the benchmarks into three classes: strings, lists, and trees.
Our set is more heavily weighted towards string benchmarks because all the
operations shared a standard string representation so we could construct many
benchmark instances by cross-checking them.
Examples of instances include checking:
\begin{enumerate}
    \item The Linux and NetBSD implementations of \texttt{strcmp} agree on all
        inputs.
    \item After inserting into an instance of Redis' linked list, using
        Redis' list-search routine to search for the item just inserted always
        succeeds.
    \item If a search for \texttt{x} in the \texttt{glibc} implementation of
        red-black trees succeeds, then after rotating a node in the tree, a
        subsequent search for \texttt{x} still succeeds.
\end{enumerate}
We tried to specialize the test harnesses to the tools' preferred format.
E.g., \tool{} expects the input to the operation to be taken as an
argument to the test harness, while the baseline tools expect this input to be
constructed by the harness itself.
Meanwhile, one of the baseline tools does not support tail recursion,
so for the benchmarks using recursion we provided it versions that were
manually transformed into a loop.
We also performed some tuning of the encodings, e.g., finding that the
baseline tools performed better when strings were allocated using
\texttt{malloc} rather than as VLAs on the stack, so we used those encodings.
We configured all tools to check only memory safety and user assertion
properties.
We have provided the full benchmark set with this submission.

\subsection{Baseline Tools}
We report comparisons against the baseline tools VeriAbsL~\citet{veriabsl},
PredatorHP~\citet{php}, and 2LS~\citet{2ls}.
We tried to represent the state-of-the-art in verification of heap-manipulating
C code, excluding tools like Astree~\citet{astree} without public executables,
but including tools like VeriAbsL that are publicly available only in binary
form.

We also considered MemCAD~\citet{memcad} and Ultimate Taipan~\citet{utaipan}.
Although it worked for small test programs, MemCAD threw many errors when we
tried to run it on our benchmark instances, apparently due to the use of C
features like initializing a struct pointer in a for loop.
Because of this, we could not run most of the benchmarks on MemCAD, and its
errors/documentation were not descriptive enough for us to adapt them in time
for this submission.
While Ultimate Taipan did properly parse and begin verifying our benchmarks, it
timed out or returned \texttt{Unknown} on all of them.
In both cases, we assume that the tools are probably tuned for different
classes of inputs and so we exclude them from our experiments and do not report
such negative results further.
%

It is also important to note that VeriAbsL, PredatorHP, and 2LS are competitors
in the SV-COMP competition, which involves detecting unsafe programs
quickly in addition to verifying safe ones.
Our evaluation considers only the verification of safe programs, as we suggest
detecting unsafety using a dedicated bounded model checking or fuzzing tool.
Hence, it should be kept in mind that these baseline models might perform
better if optimized to our setting.

\newcommand\res[2]{\diagbox[dir=NE]{#1}{#2}}
\newcommand\rnone{---}
\newcommand\rnative[1]{\res{---}{#1}}
\newcommand\rboth[1]{#1}
\begin{table}[t]
    \centering
    \setlength{\tabcolsep}{2pt}
    \begin{tabular}{lll|cccc|cc}
        \toprule
        Benchmark & Count & Kind
        & \tool{} & VeriAbsL    & PredatorHP & 2LS & Port.\ w/o & Port.\ w/ \\
        \midrule

        strings & 62 & solved & \textbf{58}    & 20        & 0          & 0 & 20 & \textbf{58} \\
        & & unique & \textbf{38}    & 0        & 0          & 0    & \\
        & & fastest & \textbf{51}    & 7        & 0          & 0    & \\

        \midrule

        lists & 26 & solved & \textbf{20} & 4 & 6 & 9 & 14 & \textbf{23} \\
        & & unique & \textbf{9} & 1 & 0 & 1 \\
        & & fastest & \textbf{9} & 2 & 6 & 6 \\

        \midrule

        trees & 17 & solved & 13 & 0 & 0 & \textbf{16} & 16 & \textbf{17} \\
        & & unique & 1 & 0 & 0 & \textbf{4} \\
        & & fastest & 3 & 0 & 0 & \textbf{14} \\
    \end{tabular}
    \caption{\label{tab:EvalTable} Evaluation Table. For each benchmark, the
    `solved' row shows how many instances that tool solved, the `unique' row shows
    how many instances were solved only by that tool, and the `fastest' row
    shows how many instances that tool solved faster than any other tool.
    ``Port.\ w/o'' shows the number solved by a virtual-best portfolio of all tools
    except \tool{}, while ``Port.\ w/'' shows the number solved by a
    virtual-best portfolio including \tool{}.}
\end{table}

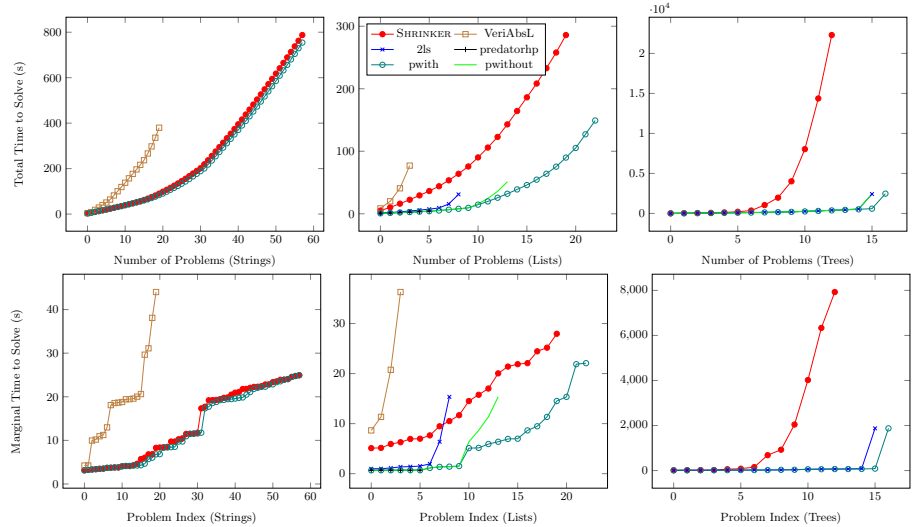
\begin{figure}[t]
    \centering
    \begin{tikzpicture}[scale=0.5]
        \begin{axis}[xlabel=Number of Problems (Strings),ylabel=Total Time to Solve (s),
                     legend pos=north west]
            \addplot[mark=*,color=red] table [col sep=comma]
                {data/processed/strings.180min.shrinker.seq.csv};
            \addplot[mark=square,color=brown] table [col sep=comma]
                {data/processed/strings.180min.veriabsl.seq.csv};
            \addplot[mark=o,color=teal] table [col sep=comma]
                {data/processed/strings.180min.pwith.seq.csv};
        \end{axis}
    \end{tikzpicture}
    \begin{tikzpicture}[scale=0.5]
        \begin{axis}[xlabel=Number of Problems (Lists),
                     legend pos=north west,legend columns=2]
            \addplot[mark=*,color=red] table [col sep=comma]
                {data/processed/lists.180min.shrinker.seq.csv};
            \addplot[mark=square,color=brown] table [col sep=comma]
                {data/processed/lists.180min.veriabsl.seq.csv};
            \addplot[mark=x,color=blue] table [col sep=comma]
                {data/processed/lists.180min.2ls.seq.csv};
            \addplot[mark=+,color=black] table [col sep=comma]
                {data/processed/lists.180min.predatorhp.seq.csv};
            \addplot[mark=o,color=teal] table [col sep=comma]
                {data/processed/lists.180min.pwith.seq.csv};
            \addplot[mark=.,color=green] table [col sep=comma]
                {data/processed/lists.180min.pwithout.seq.csv};
            \legend{\tool{},VeriAbsL,2ls,predatorhp,pwith,pwithout}
        \end{axis}
    \end{tikzpicture}
    \begin{tikzpicture}[scale=0.5]
        \begin{axis}[xlabel=Number of Problems (Trees),
                     legend pos=north west]
            \addplot[mark=*,color=red] table [col sep=comma]
                {data/processed/trees.180min.shrinker.seq.csv};
            \addplot[mark=x,color=blue] table [col sep=comma]
                {data/processed/trees.180min.2ls.seq.csv};
            \addplot[mark=o,color=teal] table [col sep=comma]
                {data/processed/trees.180min.pwith.seq.csv};
            \addplot[mark=.,color=green] table [col sep=comma]
                {data/processed/trees.180min.pwithout.seq.csv};
        \end{axis}
    \end{tikzpicture}
    \begin{tikzpicture}[scale=0.5]
        \begin{axis}[xlabel=Problem Index (Strings),ylabel=Marginal Time to Solve (s),
                     legend pos=north west]
            \addplot[mark=*,color=red] table [col sep=comma]
                {data/processed/strings.180min.shrinker.par.csv};
            \addplot[mark=square,color=brown] table [col sep=comma]
                {data/processed/strings.180min.veriabsl.par.csv};
            \addplot[mark=o,color=teal] table [col sep=comma]
                {data/processed/strings.180min.pwith.par.csv};
        \end{axis}
    \end{tikzpicture}
    \begin{tikzpicture}[scale=0.5]
        \begin{axis}[xlabel=Problem Index (Lists),
                     legend pos=north west,legend columns=2]
            \addplot[mark=*,color=red] table [col sep=comma]
                {data/processed/lists.180min.shrinker.par.csv};
            \addplot[mark=square,color=brown] table [col sep=comma]
                {data/processed/lists.180min.veriabsl.par.csv};
            \addplot[mark=x,color=blue] table [col sep=comma]
                {data/processed/lists.180min.2ls.par.csv};
            \addplot[mark=+,color=black] table [col sep=comma]
                {data/processed/lists.180min.predatorhp.par.csv};
            \addplot[mark=o,color=teal] table [col sep=comma]
                {data/processed/lists.180min.pwith.par.csv};
            \addplot[mark=.,color=green] table [col sep=comma]
                {data/processed/lists.180min.pwithout.par.csv};
        \end{axis}
    \end{tikzpicture}
    \begin{tikzpicture}[scale=0.5]
        \begin{axis}[xlabel=Problem Index (Trees),
                     legend pos=north west]
            \addplot[mark=*,color=red] table [col sep=comma]
                {data/processed/trees.180min.shrinker.par.csv};
            \addplot[mark=x,color=blue] table [col sep=comma]
                {data/processed/trees.180min.2ls.par.csv};
            \addplot[mark=o,color=teal] table [col sep=comma]
                {data/processed/trees.180min.pwith.par.csv};
        \end{axis}
    \end{tikzpicture}
    \caption{\label{fig:Cactus} Cactus plots.
    A point $(n, t)$ on the top row indicates the tool can solve $n$ of the
    benchmarks in $t$ total seconds.
    A point $(i, t)$ on the bottom row indicates the tool can solve the
    $i$th easiest (for it) benchmark in $t$ seconds; prefix-summing the bottom
    row gives the top row.
    In all cases, curves lower (faster) and to the right (solving more
    problems) are better.
    We also give curve corresponding to the virtual-best portfolio (i.e.,
    assuming a perfect heuristic that picks the best solver out of the four for
    that instance) both with (\texttt{pwith}) and without (\texttt{pwithout})
    \tool{} (for strings and trees, only one other tool solved any instances so
    the ``portfolio without'' line is identical to the other tool's curve).
    For both strings and lists, \tool{} on its own always solves more instances
    than any other tool, is within the same order of magnitude of time as
    the other tools (sometimes faster), and leads to significant improvements
    in the portfolio performance.
    For trees, \tool{} is considerably slower than the best tool (2ls), but its
    inclusion in the portfolio results in solving one additional benchmark.
    }
\end{figure}

\subsection{Results}
Our results are summarized in~\pref{tab:EvalTable} and visualized as cactus
plots in~\pref{fig:Cactus}.

For both the string and list benchmark classes, \tool{} is the single best
solver.
It is able to solve more than double the number of instances compared with the
second-best solver.
In fact, in both cases \tool{} is able to solve many benchmarks that were not
solved by any other baseline solver.
Furthermore, it does so in a reasonably small ($\leq$ 30s) amount of time per
benchmark.
There remain only four string instances unsolved, all of which involve
\texttt{strcat}, which is implemented in a complicated way for \tool{} to
follow (namely, the loop iterates simultaneously from the middle of the
destination array and the start of the source array).
On the whole, however, these results indicate \tool{} is particularly good at
solving monotonic string and list instances.

On tree benchmarks, \tool{} performs quite well, solving over 75\% of the tree
benchmarks while neither VeriAbsL nor PredatorHP solved any.
However, 2LS performs surprisingly well (solving all but one), hence \tool{}
takes second-place when looking at individual solvers.
\tool{} took longer to solve the tree benchmarks than the string and list
benchmarks because it performs path splitting up to a certain unrolling depth,
and tree-manipulating programs have an extra exponential blowup in the number
of paths (because there is a choice between left/right child during each
traversal).
This path splitting is not strictly required by \approach{}, but is needed for
\tool{} to increase precision in lieu of a more precise abstract domain
(see~\pref{app:TreeBad} for more discussion).
%
Nonetheless, \tool{} was able to solve the one instance left unsolved by 2LS.

In fact, for all three classes of benchmarks we find that adding \tool{} to a
virtual-best portfolio solver would allow it to solve more benchmarks than
would be possible without \tool{} (compare the ``Port. w/o'' and ``Port. w/''
columns in~\pref{tab:EvalTable} and the `pwith' and `pwithout' curves
in~\pref{fig:Cactus}).
%
%
Hence, in addition to being a compelling stand-alone solver for monotonic
string and list benchmarks, \tool{} could make a good addition to a portfolio
like VeriAbsL.

\section{Related Work}
\label{sec:Related}

\paragraph{Induction on Input Size.}
\citet{fpi,diffy} (integrated into the VeriAbsL portfolio) use rules to rewrite
array programs $P_N$ into a tail recursive form $P_N = \delta P_N; P_{N-1}$ and
prove correctness by inducting on the size of the input.
\citet{squeeze} describe a similar approach for array verification.
%
%
\Approach{} generalizes these ideas to a framework parameterized by
data type, measure of size, and abstraction.

\paragraph{Cyclic Proof Systems.}
Proof by infinite descent also forms the basis of
\emph{cyclic proof systems} \cite{cyclic1,cyclic2,cyclic3}.
Our main contribution is \approach{}, which is a general framework for
combining proof by infinite descent with abstract interpretation to form a
parameterized, general method of automatically proving properties about
imperative programs.
%

\paragraph{Abstract Interpretation.}
Traditional static analysis is formalized by abstract
interpretation~\cite{cousot77,ikos,astree}, and we described one formulation of
it in~\pref{sec:Background}.
%
Abstract interpreters use special abstractions of the
heap~\cite{seplogic,tvla,biabduction}.
\cite{tvla} introduced the summary nodes idea we adapted
in~\pref{sec:Implementation}.
Unfortunately, the space of such heap invariants is large, making discovering
them hard.
%
\Approach{} builds on this framework to verify monotonic data structure
traversals even without complicated heap abstractions.

\paragraph{Relational Verification.}
Reasoning about pairs of traces is part of the broader field of
\emph{relational program verification}~\citep{reltrends} of
\emph{hyperproperties}~\citep{hyperprops}.
The standard method for relational program verification is to reduce it to
nonrelational program verification on a \emph{product program} that simulates both
traces together~\citep{prodprogs,dat,covac}.
The approach taken in this paper is more similar to tools that verify
relational properties directly, without constructing a product
program~\citep{soundse,relse,circequiv}.



\paragraph{Ordering States.}
Partial order reduction, abstraction, bisimulation, symmetry-breaking, and
state merging all involve establishing an order on program
states~\cite{wellstructured,dillsymmetry,mcabs,odpor,dpor,rwset,viewabs,clarkeviewabs,commabs,heapwqo,datawqo}.
These methods generally require much stronger orderings on the traces, and give
much stronger guarantees.
E.g.,~\cite{wellstructured} guarantees completeness when there is a
well-ordering between reachable traces.
In particular, ordering traces according to their length, heap size, or input
size does not meet their requirements.
As their results generally apply to temporal analyses, our \approach{}
approach is orthogonal and complementary.

\paragraph{Completeness Thresholds and Small Model Properties.}
Our approach of reducing the size of crash traces  is similar the goal of
\emph{completeness thresholds} research~\citep{cts,lcts,lctsagain}.
Existing work in that area does not immediately apply to heap-manipulating C
programs.
\emph{Small model properties} are a similar notion in the automated reasoning
community~\citep{howsmall} including some results for polymorphic
programs~\citep{polyquickcheck,logquickcheck} and theories modeling the
heap~\citep{slhsmp,sllsmp,strand}.
%

\paragraph{Non-Temporal Analyses.}
%
\citet{proofsfromtests} use concrete tests to prove program properties, where
the tests are generated dynamically by a verification
engine.
\cite{rubytypes} show how to soundly infer static types from finitely
many test executions.
\cite{singlepass,uninterp} show that proving properties of
a restricted class of heap-manipulating imperative programs is decidable.
\cite{algpa} translates the program to a set of
recurrence relations and tries to find a closed-form solution implying
correctness properties.

\paragraph{Testing and BMC.}
Testing~\cite{aflpp,fuzzing,tdd,testingvsproofs,symex,dart,exe,cbmc,alloy}
cannot directly rule out the existence of bugs on inputs not tested.
Ways to pick test inputs are
known~\cite{inputcov,catpart,classtrees,codecov,pct}, and
our work can be interpreted as a method for proving, for a particular program,
that the small scope hypothesis holds~\citep{smallscopeeval,alloy}.
Bounded model checking tools can prove a property holds for \emph{every}
program trace up to a certain length~\citep{symex,cbmc}.
When used as a bug checker, \approach{} tends to report bugs quicker than BMC
because it uses an abstraction, i.e., it is allowed to report false alarms.
But in practice BMC can usually detect buggy variants of our MDSTs in a few
seconds, and gives counterexamples.
Hence the real challenge for these instances, and benefit of \tool{}, is in
proving safety.

\paragraph{Manual Program Analysis.}
Logics and tools for manually proving correctness
exist~\cite{floyd,hoare,vcc,fscq,coq}. In contrast to our fully automated
approach, manual proof tools require the programmer to
annotate the code with loop invariants. \cite{seplogic},
\citep{grass}, and~\cite{eprll} can express heap invariants.
%


\clearpage
\subsubsection{Acknowledgements}
I would like to thank the anonymous reviewers, whose suggestions have
dramatically improved the quality of the paper;
Dawson Engler, Alex Ozdemir, Clark Barrett, David K.\ Zhang, Geoff Ramseyer,
Alex Aiken, Martin Brain, and Aditya V.\ Thakur for extended discussions
affecting and/or significant feedback improving this work;
as well as Zachary Yedidia, Akshay Srivatsan, and attendees of the Stanford
software lunch, who provided helpful insights, conversations, and proofreading.
This work was generously funded via grants NSF DGE-1656518 and Stanford IOG
Research Hub 281101-1-UDCPQ 298911.

\subsubsection{Disclosure of Interests}
This work was generously funded via grants NSF DGE-1656518 and Stanford IOG
Research Hub 281101-1-UDCPQ 298911.

\bibliography{main}
\bibliographystyle{ACM-Reference-Format}

\appendix
\clearpage
\section{Additional Implementation Details}
\label{app:Tool}
This section complements~\pref{sec:Implementation} with more details on the
heuristics implemented by~\tool{}.
Many of the subsections refer directly to the heuristics and routines called
in~\pref{alg:OurAlgo}; the reader should refer to~\pref{sec:HerdAbstraction}
for the requirements placed on each one.
Below we use \texttt{pc} to refer to the program counter, i.e., what the next
instruction (line) to be executed is.

\subsection{Abstract Paths}
\label{app:AbstractPaths}
To assist in the heuristics we associate each abstract herd with a
\emph{abstract path} indicating the sequence of instructions associated with
selected states in the primary trace.
For example, an abstract herd having associated abstract path \texttt{p1, p3,
*, p7, *, p8} indicates that every primary trace of every herd in its
concretization set has \texttt{pc=p1} in its first state, then \texttt{pc=p3}
in its second, then zero or more other states (e.g., many executions of a
loop), then a state with \texttt{pc=p7}, then zero or more other states, and
finally a state with \texttt{pc=p8}.
The abstract path is used to control precision loss and widening in our
heuristics; see below for more details.
We replace loop iterations after a certain unrolling level with \texttt{*} to
ensure that the set of possible abstract paths is bounded (we may also do this
before the unrolling level is reached if a summary node is accessed by the
program; see below).
Our implementation of $\hWiden{}$ works by joining states that share the same
abstract path, hence having a larger unrolling level makes the analysis
significantly more precise.

\subsection{Scapegoat Construction ($\hMaybeAddScapegoats$)}
Recall that we need a way to add scapegoat traces to the trace.
Our strategy waits until enough steps and splits have occurred to guarantee
that the program input has size greater than some user-specified size bound $k$
(in our evaluation, we start this size bound at $k=1$ and retry with
incrementally higher bounds if it fails).
Once we know the input is large enough, we apply one of two shrinking rules to
the first state of the primary to construct a new initial trace for the
scapegoat:
\begin{enumerate}
    \item If the input is a linked structure, we construct a scapegoat trace
        formed by removing one of the first $k$ reachable nodes in the structure
        (updating the other pointers to skip over it). When $k > 1$, i.e.,
        multiple nodes are guaranteed to be reachable, we add separate
        scapegoats skipping over each of them.
    \item If the input is an array, and it has length at least 1, we add a
        scapegoat trace formed by removing exactly one of the first $k$
        elements from the array. To do so, we decrement the array's length,
        increment the pointer to the array's first element, update the base and
        bound values used for memory checking, and then move elements among the
        first $k$ to simulate deleting the desired element. Once again, we add
        a separate scapegoat for each of the $k$ array elements that we
        consider deleting.
\end{enumerate}
The resulting scapegoat traces start off with only one, length-1 trace.
Since dropping the head of an array or skipping a node in a linked structure
still results in a valid inital trace, these operations satisfy the constraints
on $\hMaybeAddScapegoats$.
Notably, all of this happens in the abstract, i.e., rather than constructing a
concrete trace we add all of the constraints that would result from such a
construction to our list of constraints.

\subsection{Abstract Herd Splitting ($\hSplit$)}
After a branch operation on the primary trace, $\hSplit$ is used to partition
the abstract trace herd into separate abstract trace herds representing each
possible branch outcome.
For example, if we encounter a branch that goes to \texttt{pc=p5} when $x=0$ or
\texttt{pc=p8} otherwise, we will duplicate the abstract trace herd $\pur{a'}$
into (1) $\pur{a'_1}$ with the additional constraints ``$x = 0$ and
\texttt{pc=p5} in the most recent state'' and (2) $\pur{a'_2}$ with the
additional constraints ``$x \neq 0$ and \texttt{pc=p8} in the most recent
state.''
In this way, we always know the exact next-to-be-executed instruction before
calling $\hStepSharp$.

Furthermore, if the next instruction will access a pointer in the program
(e.g., by dereferencing it or comparing it against another pointer), we also
split on the possible memory locations that the pointer could point to.
In this way, $\hStepSharp$ always knows exactly which memory location is being
referred to on all pointer operations.

\subsection{Stepping Operations ($\hStepSharp$, $\hStepE$)}
We implement $\hStepSharp$ and $\hStepE$ by adding new constraints that relate
values in a new last state of the trace to those in the previously last (now
second-to-last) state.
Every constraint previously referring to, e.g., ``the second-to-last state'' is
updated to refer to ``the third-to-last state,'' and new constraints are added
to define the now-last state.
For $\hStepE$ (which is used to advance the non-primary, scapegoat traces) we
drop any scapegoat traces where the \texttt{pc} on the final state of the trace
is unknown.
Note that $\hSplit$ is used to enforce a similar behavior for $\hStepSharp$,
i.e., the primary trace.
Thus for every abstract trace herd in the worklist, we know the next
instruction to be executed for all of the traces in every herd in the
concretization set.

We allow the user to request nondeterministic values.
$\hStepSharp$ implements this by asserting equality between the output register
and a fresh (unconstrained except for type bounds) variable.
Recall, however, that $\hStepE$ is allowed to guess nondeterministic values.
To do so, it identifies the most similar state in the primary trace (using a
heuristic based on local variables described below) and asserts that the output
register in the scapegoat trace takes on the same value that was returned in
that step of the primary trace.

\subsection{Stepper Heuristic}
Recall at each step we must determine how far to advance each of the scapegoat
traces.
We annotate all the loops with `entrance,' `iteration,' and `end' instructions,
and associate each loop with the set of `relevant local variables,' i.e., those
that it may write to in the loop body (e.g., an array iteration might write to
a counter \texttt{i}, a list traversal might write to a pointer \texttt{l},
etc.).
We only ever advance scapegoat traces after the primary has executed an
`iteration' instruction.
Then we apply $\hStepE$ repeatedly to advance the scapegoat trace until we can
prove that it (1) reaches the same iteration instruction and (2) the relevant local
variables in the primary trace have the same values as those in the scapegoat
trace (a difference of 1 is allowed for integer variables to account for the
fact that we have have shrunk the input's size by one).
If a scapegoat trace ever reaches a branch instruction where the branch to be
taken is not already implied by existing constraints, we remove that scapegoat
trace from the herd.


\subsection{State Querying ($\hCanFail$, $\hCanBlame$, $\hMorePrecise$)}
$\hCanFail$ and $\hCanBlame$ are implemented by checking whether the set of
constraints implies the conditions needed: $\hCanFail$ checks whether the
program counter on the final state could be a failure operation, while
$\hCanBlame$ checks whether the constraints imply both that the scapegoat trace
has a smaller size and definitely has the same failing program counter.
$\hMorePrecise(\pur{a}, \pur{b})$ is implemented by checking whether every
constraint in $\pur{b}$ is also in $\pur{a}$.

\subsection{Dataflow Optimizations}
To minimize the amount of analysis that needs to be done by our
\approach{}-based analysis, we apply simple dataflow-based optimizations first
to, e.g., elide obviously duplicate checks, eliminate common subexpressions,
and delete dead code.
These optimizations are carefully designed to ensure that they never remove a
bug from the program, i.e., we can only remove a pointer validity check if we
know that it would only fail if some earlier pointer validity check would have
failed before reaching it.

\subsection{Parallelization}
\label{app:Parallel}
One major optimization we applied to~\pref{alg:OurAlgo} was to parallelize the
main loop.
We start with one worker process that is running the algorithm in a sequential
manner as described in~\pref{sec:Theory}.
When fewer than some user-specified maximum number of worker processes are
running, worker processes will attempt to split their worklists in two to use
the idle machine cores; the two resulting workers perform~\pref{alg:OurAlgo} as
normal, but only on their own half of the original worker's worklist.
This parallelization of the main loop is sound~(\pref{lem:OurAlgo} still holds,
guaranteeing that \emph{at least one worker} processed an abstract herd
accounting for any given reachable trace), but can introduce nondeterminism if
done na\"ively (because $\hWiden$ may rely on the other elements in the
worklist when deciding how much to widen).
To avoid this, we only split off work when we can guarantee that $\hWiden$ in
one partition would never use the elements in the other partition; because
$\hWiden$ joins only elements with the same abstract path, this corresponds to
checking that their paths all have disjoint prefixes.

\subsection{Measure of Trace Size}
\label{app:TraceSize}
\Approach{} requires the analysis designer to specify the measure of trace
size.
The tool is sound regardless of this choice; it only affects completeness.
Ideally, the measure ensures that scapegoat traces (as added by
$\hMaybeAddScapegoats$, which in \tool{} drops one item from the input
structure) are smaller than their primary.
\tool{} takes the size to be the number of allocated memory regions reachable
from the input arguments plus the number of elements in arrays reachable from
the input arguments plus the number of times \texttt{malloc} was called.
This captures the number of items allocated on the heap, and does indeed get
smaller as we drop elements from input structures.
Otherwise, not much thought went into this choice.
We have every reason to expect the tool could perform as well with many
different notions of size, such as the number of bytes allocated on the heap or
the length of the trace.
One practical benefit of counting the number of allocation regions rather than
raw number of bytes allocated or number of IR instructions executed is that it
was more interpretable when debugging \tool{}: dropping a single node in a
linked list input structure only changes the number of allocated regions by 1,
but it changes the raw number of allocated bytes by the difficult-to-eyeball
quantity \texttt{sizeof(struct list\_node)} and changes the number of executed
instructions by an even harder-to-predict number.

\subsection{Disjoint Heaps Only Applies to Harness}
\label{app:DisjointHeapsInterior}
In~\pref{sec:Implementation} we described how the program is verified under the
assumption that the test harness is called with inputs that point to disjoint
heaps.
But this assumption \emph{does not} apply to calls made by the test harness.
For example, the following program test harness is allowed by \tool{} and
correctly checks whether \texttt{copy\_ints} correctly handles overlapping
source and destination regions (i.e., \texttt{memmove} vs.\ \texttt{memcpy}
semantics).
\begin{minted}{C}
// ... eq_arrays, copy_ints assumed to be defined here ...
struct array { int *data; unsigned n_data; }
void test(struct array A, struct array B) {
    __VERIFIER_assume(eq_arrays(A, B));
    __VERIFIER_assume(A.n_data >= 2);

    // The aliasing here is allowed by SHRINKER
    copy_ints(A.data + 1, A.data, A.n_data - 1);

    for (unsigned i = 1; i < B.n_data; i++)
        __VERIFIER_assert(A.data[i] == B.data[i - 1]);
}
\end{minted}

\subsection{Array Types}
\label{app:ArrayTypes}
To simplify parsing, \tool{} rejects programs that declare array-typed values
in C.
In general, arrays are second-class types in C (e.g., they automatically decay
to pointers in most contexts) and can usually be replaced with pointers.
For example, \tool{} would reject the following program, which declares the
variable \texttt{string} having an explicit array type:
\begin{minted}{C}
void test(unsigned n) {
    char string[n];
    // ... "string" is an array of "n" chars ...
}
\end{minted}
Recall from~\pref{sec:Implementation} that parameters to the test harness
having a struct type with a pointer field \texttt{X} and integer field
\texttt{n\_X} are interpreted by \tool{} as arrays.
So the above program can instead be rewritten to use pointers as below, which
accomplishes the original intention and will be accepted by \tool{}.
\begin{minted}{C}
struct string { char *string; unsigned n_string; }
void test(struct string string) {
    // ... "string.string" points to "string.n_string" chars ...
}
\end{minted}

\subsection{Reasoning About Array Equality}
In order to effectively perform \approach{}, \tool{} needs to be able to track
that certain arrays are identical between the primary and scapegoat traces
(e.g., that an array in the sacpegoat trace is equal to the corresponding array
in the primary with its first element removed).
We do this by encoding arrays as uninterpreted objects, which allows us to
track equality constraints through the application of other uninterpreted
functions like \texttt{store} and \texttt{select}.
For example, we might have the following constraint, stating that some array is
identical between states 10 of the primary and scapegoat traces:
\begin{verbatim}
array_1_in_primary_state_10 = array_1_in_scapegoat_state_10
\end{verbatim}
Then, suppose we apply $\hStepSharp$ and $\hStepE$ to add new constraints
stepping each trace forward once.
If this involves writing a value \texttt{V} to index \texttt{K} of each array,
the constraints would look like:
\begin{verbatim}
array_1_in_primary_state_10 = array_1_in_scapegoat_state_10
array_1_in_primary_state_11
    = store(array_1_in_primary_state_10, K, V)
array_1_in_scapegoat_state_11
    = store(array_1_in_scapegoat_state_10, K, V)
\end{verbatim}
Then our abstract domain implementation, which propagates uninterpreted
function equalities, can conclude from this that the array is still equal in
state 11 of the primary and scapegoats.
\begin{verbatim}
array_1_in_primary_state_11 = array_1_in_scapegoat_state_11
\end{verbatim}

\section{Extended Worked Example}
\label{app:Worked}
We now work through a small example showing how \tool{} can use
\pref{alg:OurAlgo} to prove correctness of a simple heap-manipulating program.
In the below we will say an abstract herd $\pur{a}$ ``represents herds \ldots''
to mean the concretization set $\hgamma(\pur{a})$ consists of such herds.
Recall our running example from~\pref{sec:Motivating}:
\begin{minted}{C}
struct arr { int *data; int n_data; };
void test(struct arr arr) {
    for (int i = 0; i < arr.n_data; i++)
        arr.data[i] = 0;
    for (int i = 0; i < arr.n_data; i++)
         if (arr.data[i] != 0)
             __VERIFIER_fail(); }
\end{minted}
Recall this is an unusually simple example for the sake of exposition; our
actual benchmark instances are more complicated.
Furthermore, the actual execution of \tool{} works at a very low level
(essentially tracking the values of dozens of registers in addition to base and
bound arrays for memory checking instrumentation after lowering this code), so
we have tried for exposition reasons to present the intermediate states in a
relatively succinct way.
In particular, in the example below we assume the program transition is very
coarse-grained (e.g., processes the entirety of line 4 in a single step).

Furthermore, instead of showing the intermediate steps of the algorithm, we
only show the resulting proof, i.e., a set of abstract herds that are
\emph{closed}, i.e., if $\blue{t}$ is a trace accounted for by one of the
abstract herds in the set, and $\blue{t'} \in \cStep(\blue{t})$, then
$\blue{t'}$ is also accounted for by some abstract herd in the set.
In other words, these abstract herds can be thought of as the contents of the
\texttt{seen} set in~\pref{alg:OurAlgo}.
The relation can be visualized in the following graph, where each node is one
of the abstract herds, and if $\blue{t}$ is accounted for by some node in the
graph and $\blue{t'} \in \cStep(\blue{t})$, then $\blue{t'}$ is accounted for
by one of the successor nodes in the graph.

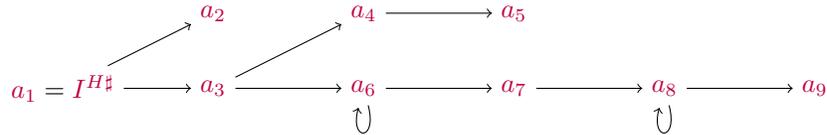
\begin{center}
    \begin{tikzpicture}
        \node (A1) at (0, 0) {$\pur{a_1} = \pur{I^{H\sharp}}$};
        \node (A2) at (2, 1) {$\pur{a_2}$};
        \node (A3) at (2, 0) {$\pur{a_3}$};
        \node (A4) at (4, 1) {$\pur{a_4}$};
        \node (A5) at (6, 1) {$\pur{a_5}$};
        \node (A6) at (4, 0) {$\pur{a_6}$};
        \node (A7) at (6, 0) {$\pur{a_7}$};
        \node (A8) at (8, 0) {$\pur{a_8}$};
        \node (A9) at (10, 0) {$\pur{a_9}$};

        \draw[->] (A1) -- (A2);
        \draw[->] (A1) -- (A3);
        \draw[->] (A3) -- (A4);
        \draw[->] (A4) -- (A5);
        \draw[->] (A3) -- (A6);
        \path[->] (A6) edge [loop below] (A6);
        \draw[->] (A6) -- (A7);
        \draw[->] (A7) -- (A8);
        \path[->] (A8) edge [loop below] (A8);
        \draw[->] (A8) -- (A9);
    \end{tikzpicture}
\end{center}

Some structure of the program is visible in this graph, e.g., $\pur{a_2}$
corresponds to the program exiting quickly if \texttt{arr.n\_data = 0}, while
$\pur{a_4}$ and $\pur{a_5}$ correspond to the program exiting after only one
iteration of each loop if \texttt{arr.n\_data = 1}.
The remaining abstract herds capture only traces where \texttt{arr.n\_data >
1}: $\pur{a_6}$ captures subsequent iterations of the first loop, $\pur{a_7}$
captures the first iteration of the second loop, $\pur{a_8}$ captures
subsequent iterations of the second loop, and $\pur{a_9}$ captures any
iterations of the second loop that might reach the
\texttt{\_\_VERIFIER\_fail()} statement.
Notably, the analysis knows easily that it is not possible for $\pur{a_7}$
(first iteration of the second loop) to reach failure because it tracks the
possible values of \texttt{arr.data[0]}; but it does not do the same for
$\pur{a_9}$ (subsequent iterations of the second loop) because that would
require tracking the arbitrarily many possible values in the rest of
\texttt{arr.data[1:n]}.
Because of this, the analyzer was not able to rule out the possibility of
$\pur{a_9}$ without the use of \approach{}, which tells it that blame can be
transferred onto the scapegoat trace in $\pur{a_9}$, hence no error need be
reported there.

We now describe each of the abstract herds.
Recall that an abstract herd can be expressed via constraints, where the
concretization set is all of the herds that satisfy those constraints.
In the below, we assume every state has a \texttt{pc} indicating the line that
is about to execute next.
Lines 3 and 5 indicate checking the \texttt{i < arr.n\_data} condition.
Some constraints need to relate the values between different traces; for this
we write $\hidx{h}{1}(\mathtt{foo})$ to mean ``the value of $\mathtt{foo}$ in
the last state of trace $\hidx{h}{1}$.''
We also use \texttt{arr.data[i:j]} to mean ``the subarray pointed to by
\texttt{arr.data} from index \texttt{i} (inclusive) to index \texttt{j}
(exclusive).''

\begin{itemize}
    \item $\pur{a_1} = \pur{I^{H\sharp}}$: (initial state)
        \begin{itemize}
            \item $\hidx{h}{1}$ has length 1:
                \begin{itemize}
                    \item First state: \texttt{arr.n\_data >= 0}, \texttt{i =
                        0}, \texttt{pc = 3}.
                \end{itemize}
        \end{itemize}
        Successors: $\pur{a_2}$ (empty), $\pur{a_3}$ (nonempty).

    \item $\pur{a_2}$: (empty array)
        \begin{itemize}
            \item $\hidx{h}{1}$ has length 2:
            \begin{itemize}
                    \item First state: \texttt{arr.n\_data = 0},
                        \texttt{i = 0}, \texttt{pc = 3}.
                    \item Second state: identical except \texttt{pc = 5}.
            \end{itemize}
        \end{itemize}
        Successors: none (constraints imply the primary trace reaches exit
        immediately after this state).

    \item $\pur{a_3}$: (finished first iteration of first loop)
        \begin{itemize}
            \item $\hidx{h}{1}$ has length 2:
                \begin{itemize}
                    \item First state: \texttt{arr.n\_data >= 1},
                        \texttt{i = 0}, \texttt{pc = 3}.
                    \item Second state: identical except \texttt{arr.data[0] =
                        0}, \texttt{i = 1}.
                \end{itemize}
            \item $\hidx{h}{2}$ has length 1:
                \begin{itemize}
                    \item First state: $\mathtt{arr.n\_data} =
                        \hidx{h}{1}(\mathtt{arr.n\_data}-1)$, $\mathtt{pc =
                        3}$, \texttt{i = 0}, and
                        $\mathtt{arr.data[0:\mathtt{arr.n\_data}]} =
                        \hidx{h}{1}(\mathtt{arr.data}[1:\mathtt{arr.n\_data}])$.
                \end{itemize}
        \end{itemize}
        Successors: $\pur{a_4}$ (finished), $\pur{a_6}$ (unfinished).

    \item $\pur{a_4}$: (array size exactly 1)
        \begin{itemize}
            \item $\hidx{h}{1}$ has length 3:
                \begin{itemize}
                    \item First state: \texttt{arr.n\_data = 1}, \texttt{i =
                        0}, \texttt{pc = 3}.
                    \item Second state: identical except \texttt{arr.data[0] =
                        0}, \texttt{i = 1}.
                    \item Third state: identical except \texttt{i = 0},
                        \texttt{pc = 5}.
                \end{itemize}
            \item ($\hidx{h}{2}$ gets dropped)
        \end{itemize}
        Successors: $\pur{a_5}$

    \item $\pur{a_5}$: (array size exactly 1)
        \begin{itemize}
            \item $\hidx{h}{1}$ has length 4:
                \begin{itemize}
                    \item First state: \texttt{arr.n\_data = 1}, \texttt{i =
                        0}, \texttt{pc = 3}.
                    \item Second state: identical except \texttt{arr.data[0] =
                        0}, \texttt{i = 1}.
                    \item Third state: identical except \texttt{i = 0},
                        \texttt{pc = 5}.
                    \item Fourth state: identical except \texttt{i = 1},
                        \texttt{pc = 5}.
                \end{itemize}
        \end{itemize}
        Successors: (none; the program immediately exits after this.)

    \item $\pur{a_6}$: (>1 iterations of the first loop)
        \begin{itemize}
            \item $\hidx{h}{1}$ has length $\geq 3$:
                \begin{itemize}
                    \item (Earlier states unconstrained)
                    \item Final state: \texttt{arr.n\_data >= 2},
                        \texttt{arr.data[0] = 0}, \texttt{i >= 2}, \texttt{pc =
                        3}.
                \end{itemize}
            \item $\hidx{h}{2}$ has length $\geq 2$:
                \begin{itemize}
                    \item (Early states unconstrained)
                    \item Last state: $\mathtt{arr.n\_data} =
                        \hidx{h}{1}(\mathtt{arr.n\_data}-1)$, $\mathtt{pc =
                        3}$, $\mathtt{i} = \hidx{h}{1}(\mathtt{i} - 1)$, and
                        $\mathtt{arr.data[0:\mathtt{arr.n\_data}]} =
                        \hidx{h}{1}(\mathtt{arr.data}[1:\mathtt{arr.n\_data}])$.
                \end{itemize}
        \end{itemize}
        Successors: $\pur{a_7}$ (finished), $\pur{a_6}$ (unfinished).

    \item $\pur{a_7}$: (start of second loop, after >1 iterations of the first loop)
        \begin{itemize}
            \item $\hidx{h}{1}$ has length $\geq 3$:
                \begin{itemize}
                    \item (Earlier states unconstrained)
                    \item Final state: \texttt{arr.n\_data >= 2},
                        \texttt{arr.data[0] = 0}, \texttt{i = 0}, \texttt{pc =
                        5}.
                \end{itemize}
            \item $\hidx{h}{2}$ has length $\geq 2$:
                \begin{itemize}
                    \item (Early states unconstrained)
                    \item Last state: $\mathtt{arr.n\_data} =
                        \hidx{h}{1}(\mathtt{arr.n\_data}-1)$, $\mathtt{pc =
                        5}$, $\mathtt{i} = 0$, and
                        $\mathtt{arr.data[0:\mathtt{arr.n\_data}]} =
                        \hidx{h}{1}(\mathtt{arr.data}[1:\mathtt{arr.n\_data}])$.
                \end{itemize}
        \end{itemize}
        Successors: $\pur{a_8}$ (only step the primary forward; failure not
        possible in the primary because we know \texttt{arr.data[0] = 0})

    \item $\pur{a_8}$: (>1 iterations of second loop, after >1 iterations of
        the first loop)
        \begin{itemize}
            \item $\hidx{h}{1}$ has length $\geq 3$:
                \begin{itemize}
                    \item (Earlier states unconstrained)
                    \item Final state: \texttt{arr.n\_data >= 2},
                        \texttt{arr.data[0] = 0}, \texttt{i >= 1}, \texttt{pc =
                        5}.
                \end{itemize}
            \item $\hidx{h}{2}$ has length $\geq 2$:
                \begin{itemize}
                    \item (Early states unconstrained)
                    \item Last state: $\mathtt{arr.n\_data} =
                        \hidx{h}{1}(\mathtt{arr.n\_data}-1)$, $\mathtt{pc =
                        5}$, $\mathtt{i} = \hidx{h}{1}(\mathtt{i} - 1)$, and
                        $\mathtt{arr.data[0:\mathtt{arr.n\_data}]} =
                        \hidx{h}{1}(\mathtt{arr.data}[1:\mathtt{arr.n\_data}])$.
                \end{itemize}
        \end{itemize}
        Successors: $\pur{a_8}$ (no failure, unfinished), $\pur{a_9}$
        (failure), (the final branch with \texttt{i=arr.n\_data} results in
        immediate program exit, not shown).

    \item $\pur{a_9}$: (failure after >1 iterations of second loop)
        \begin{itemize}
            \item $\hidx{h}{1}$ has length $\geq 3$:
                \begin{itemize}
                    \item (Earlier states unconstrained)
                    \item Final state: \texttt{arr.n\_data >= 2},
                        \texttt{arr.data[0] = 0}, \texttt{i >= 1}, \texttt{pc =
                        7}.
                \end{itemize}
            \item $\hidx{h}{2}$ has length $\geq 2$:
                \begin{itemize}
                    \item (Early states unconstrained)
                    \item Last state: $\mathtt{arr.n\_data} =
                        \hidx{h}{1}(\mathtt{arr.n\_data}-1)$, $\mathtt{pc =
                        7}$, $\mathtt{i} = \hidx{h}{1}(\mathtt{i} - 1)$, and
                        $\mathtt{arr.data[0:\mathtt{arr.n\_data}]} =
                        \hidx{h}{1}(\mathtt{arr.data}[1:\mathtt{arr.n\_data}])$.
                \end{itemize}
        \end{itemize}
        Note that both have reached failure, so $\hCanBlame$ is true and we do
        not need to report a potential failure.

        Successors: (none; after failure the program halts)
\end{itemize}

\section{Limitations and Future Work}
\label{app:Limitations}
We now discuss some major limitations of and future work for \approach{} in
general and \tool{} in particular.

\subsection{Performance on non-MDST Instances}
\tool{} and \approach{} are designed to take advantage of the fact that many
real-world MDSTs do very similar things when run on an arbitrary input as when
run on a shrunk version of that input.
When analyzing programs that do not have such a property, the technique
essentially reduces to traditional abstract interpretation~(\pref{alg:AbsInt}),
where the analysis power is controlled directly by the precision of the
abstract domain.
\tool{} does not use a particularly precise abstraction, so we do not expect or
claim it to work well on such non-MDST programs.
Determining whether the key insight of \approach{} can be useful in non-MDST
settings is an interesting area of future work.

\subsection{Nested Loops}
This paper only considered `singly nested' MDSTs, excluding, e.g.,
lists-of-strings that might also be traversed in a monotonic way, and we make
no claims about the performance of \tool{} on nested structures.
To support such nested MDSTs we would need to extend our heuristics for
stepping and adding scapegoats to handle such cases.
Alternatively, we could try verifying them in a compositional way, i.e.,
analyze just the inner structure first to determine lemmas about its behavior
that allow us to then analyze the outer structure independently.
These are interesting areas of future work but beyond the scope of this paper.

\subsection{Skip Traversals}
We have focused on MDSTs that iterate forward by a single element on each
iteration, but one could imagine MDSTs that move forward by a different
constant (or even variable) number of elements each iteration.
For example, searching in an array of integers where every pair of two adjacent
integers are considered a single key--value pair.
To profitably apply \tool{} to such programs, we would need to modify its
implementation of $\hMaybeAddScapegoats$ to drop more than one entry when
creating the scapegoat trace; in the earlier example, dropping the first two
entries (i.e., the first logical key--value pair) would suffice.

\subsection{More Precise Memory Model}
To simplify analyses, \tool{} currently rejects (reports unknown on) any
program that tries to reinterpret, say, a pointer to a struct as a pointer to a
different kind of struct or manipulate its byte value.
This disallows generic operations like \texttt{memcpy} that interpret arrays as
byte pointers, and intrusive generic data structures where a pointer to a
struct's field is subtracted from to get a pointer to the struct itself.
This is \emph{not} a fundamental limitation of \approach{}, and we believe that
our memory abstraction can be extended to support many such common operations.
The simplest way to add some support for such byte-level operations is to
detect them and reinterpret them in terms of our higher level representation,
e.g., we can detect when a constant number of bytes is subtracted from a
pointer and then interpret that as making it point to the corresponding
container element.
With more engineering effort, a more complete solution would involve modifying
\tool{} to use a byte-level abstraction of the heap, where every node in the
heap is subdivided into individual bytes that can be pointed to, read from, and
written to individually.
We believe that this would not be too difficult to implement within \tool{},
although tracking relations between individual bytes might slow down the
numerical domain reasoning.

\subsection{Performance on Tree Instances}
\label{app:TreeBad}
In our evaluation we saw that tree instances were particularly slow for \tool{}
to verify.
This is because \tool{} relies heavily on path splitting to keep the trace herd
abstractions precise, and tree traversals have an exponential blowup in the
number of possible paths because at each node you can go either left or right.
There is nothing in the theory of \approach{} that requires such aggressive
path splitting; hypothetically, $\hWiden$ could even be implemented to join the
abstract trace herd with all previously seen abstract trace herds, at the cost
of precision.

However, path splitting \emph{is} required for the current version of \tool{},
with its current heap abstraction and heuristics, to solve most the tree
benchmarks.
This is because \tool{} needs relatively precise information about the tree and
the primary trace's path in order to pick what scapegoats to add and prove that
they can be blamed (i.e., that they do ``essentially the same thing'' as the
primary).
For example, consider verifying a BST search routine, and suppose we know that
the primary trace went down the rightmost branch of the tree:
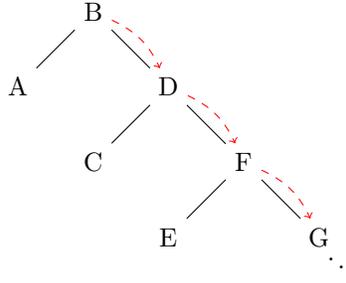
\begin{center}
    \begin{tikzpicture}
        \node (A) at (0, 0) {B};
        \node (B) at (-1, -1) {A};
        \node (C) at (1, -1) {D};
        \node (D) at (0, -2) {C};
        \node (E) at (2, -2) {F};
        \node (F) at (1, -3) {E};
        \node (G) at (3, -3) {G};
        \node (H) at (3.3, -3.3) {$\ddots$};
        \draw (A) -- (B);
        \draw (A) -- (C);
        \draw (C) -- (D);
        \draw (C) -- (E);
        \draw (E) -- (F);
        \draw (E) -- (G);

        \draw[->,red,dashed] (A) to[bend left=20] (C);
        \draw[->,red,dashed] (C) to[bend left=20] (E);
        \draw[->,red,dashed] (E) to[bend left=20] (G);
    \end{tikzpicture}
\end{center}
With this information, \tool{} can determine that the scapegoat trace resulting
from an input tree formed by dropping the right child of the root (and
replacing it with \emph{its} right child) results in similar enough behavior to
complete the \approach{}-based verification.
This is because it will take the same traversal, just skipping over the
iteration that would have touched the right child of the root:
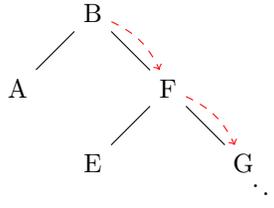
\begin{center}
    \begin{tikzpicture}
        \node (A) at (0, 0) {B};
        \node (B) at (-1, -1) {A};
        \node (E) at (1, -1) {F};
        \node (F) at (0, -2) {E};
        \node (G) at (2, -2) {G};
        \node (H) at (2.3, -2.3) {$\ddots$};
        \draw (A) -- (B);
        \draw (A) -- (E);
        \draw (E) -- (F);
        \draw (E) -- (G);

        \draw[->,red,dashed] (A) to[bend left=20] (E);
        \draw[->,red,dashed] (E) to[bend left=20] (G);
    \end{tikzpicture}
\end{center}
Crucially, however, for herds where the primary trace takes a different path
through the tree, different scapegoats might be nedeed.
If the primary instead took the path B, D, C, \ldots, we would instead have to
replace D with its left child C, not its right child F.
For this reason, \tool{} relies on aggressive path splitting so that for every
abstract trace herd it processes, it knows enough about the primary trace's
path to add the proper scapegoat and prove that its trace is similar enough to
guarantee it can accept blame if the primary fails.

In theory, and with more engineering effort, a more precise abstract herd
domain might be developed that could obviate the need to do such aggressive
path splitting by representing parameterized constraints, e.g., that the
scapegoat trace is the result of dropping \emph{some} node on the primary
trace's traversal path, but keep which node exactly that is symbolic.
Unfortunately, we believe such an abstract domain would be significantly more
complicated to implement, and would work against the main benefit of
\approach{}, i.e., its ability to use simpler abstract domains to represent the
heap.

%

\subsection{More Precise Numerical Domains}
We used a custom integer difference logic (IDL) solver for the core of our
numerical abstraction~(\pref{sec:NumericalAbstraction}).
We did this to keep the tool self-contained, easy-to-build, and have a small
trusted computing base (TCB).
But there are already implementations of many abstract domains, including
octagons (similar to IDL), in production quality libraries like
APRON~\cite{apron}.
We could modify \tool{} to use a library like APRON, which would probably
improve performance and let us opt-in to more precise abstract domains as
desired, at the cost of our TCB size and perhaps the ease-of-use issue of
adding dependencies.

\subsection{Connections to the Small Scope Hypothesis}
The original motivation of this work was to better understand the \emph{small
scope hypothesis}.
The observation is that many programs \emph{feel}, intuitively, like they are
either correct, or they fail on some small input.
It strains credibility that a tiny, 5-line linked list search-and-delete
routine could be correct for all inputs up to size 1007, but fail on an input
of size 1008.
But our understanding of program verification says very little about why this
feeling should be justified.

The scapegoating size descent analysis technique sheds some light on this
mystery, because it works by proving that failure-inducing inputs are small.
Essentially, it leads to an explanation of the small scope hypothesis for
programs that do `essentially the same thing' on a smaller version of the
input.
Ultimately, we would like to adapt the results in this paper into a
\emph{syntactic} result of the form: any program in this syntactic class
satisfies the small scope hypothesis, i.e., is correct on all inputs if and
only if it is correct on all inputs of a certain size.

\end{document}